\documentclass[11pt,english]{article}

\usepackage[T1]{fontenc}
\usepackage[latin9]{inputenc}
\usepackage{geometry}
\usepackage{verbatim}
\usepackage{float}
\usepackage{amsmath}
\usepackage{amsthm}
\usepackage{amssymb}
\usepackage{graphicx}
\usepackage{lmodern}
\usepackage{babel}
\usepackage{enumitem}
\usepackage{array}
\usepackage{wrapfig}

\usepackage[noend]{algpseudocode}

\usepackage[procnumbered,ruled,vlined,linesnumbered]{algorithm2e}
\usepackage{xcolor}
\usepackage{nameref}
\definecolor{ForestGreen}{rgb}{0.1333,0.5451,0.1333}
\definecolor{DarkRed}{rgb}{0.8,0,0}
\definecolor{Red}{rgb}{1,0,0}
\usepackage[linktocpage=true,
pagebackref=true,colorlinks,
linkcolor=DarkRed,citecolor=ForestGreen,
bookmarks,bookmarksopen,bookmarksnumbered]
{hyperref}

\usepackage{cleveref}

\usepackage{mathtools}
\usepackage{thmtools}

\declaretheorem[numberwithin=section]{theorem}
\declaretheorem[numberlike=theorem]{lemma}
\declaretheorem[numberlike=theorem,name=Lemma]{lem}

\declaretheorem[numberlike=theorem,name=Proposition]{prop}
\declaretheorem[numberlike=theorem]{corollary}
\declaretheorem[numberlike=theorem,name=Corollary]{cor}
\declaretheorem[numberlike=theorem,style=definition]{definition}
\declaretheorem[numberlike=theorem,style=definition,name=Definition]{defn}
\declaretheorem[numberlike=theorem]{claim}

\declaretheorem[numberlike=theorem,style=remark]{remark}

\theoremstyle{remark}
\newtheorem*{rem*}{Remark}

\newcommand{\alert}[1]{\textcolor{red}{#1}}
\newcommand{\eat}[1]{}

\def\ShowComment{True} 

\ifdefined\ShowComment
\def\thatchapholtext#1{\textcolor{purple}{#1}}
\def\thatchaphol#1{\marginpar{$\leftarrow$\fbox{T}}\footnote{$\Rightarrow$~{\sf\textcolor{purple}{#1 --Thatchaphol}}}}
\def\danupon#1{\textcolor{orange}{DN: #1}}
\def\sorrachai#1{\marginpar{$\leftarrow$\fbox{S}}\footnote{$\Rightarrow$~{\sf\textcolor{green}{#1
-- Sorrachai}}}}
\def\jason#1{\marginpar{$\leftarrow$\fbox{J}}\footnote{$\Rightarrow$~{\sf\textcolor{blue}{#1 --Jason}}}}
\def\note#1{#1}
\def\alert#1{\textcolor{red}{#1}}

\else
\def\thatchapholtext#1{}
\def\thatchaphol#1{}
\def\danupon#1{}
\def\shen#1{}
\def\sorrachai#1{}
\def\jason#1{}
\def\note#1{} 
\def\alert#1{}

\fi

\global\long\def\cP{\mathcal{P}}

\global\long\def\phipart{\phi_{\mathrm{part}}}

\global\long\def\polylog{\mathrm{polylog}}

\global\long\def\mincut{\mathrm{mincut}}

\global\long\def\path{\mathrm{path}}

\global\long\def\Otil{\tilde{O}}

\global\long\def\Sparsify{\textsc{Sparsify}}
\global\long\def\Partition{\textsc{Partition}}
\global\long\def\Refine{\textsc{Refine}}
\global\long\def\RefineOnSparsifier{\textsc{Refine}}
\global\long\def\PartialTree{\textsc{SmallConn}}
\global\long\def\GHtree{\textsc{GHtree}}
\global\long\def\GHtreeSparse{\textsc{GHtreeSparse}}
\global\long\def\ssc{$k$-\textsc{SSC}\xspace}
\global\long\def\sscv{$k$-\textsc{SSC Verification}\xspace}
\global\long\def\rep{\mathrm{rep}}

\newcommand{\lar}{{\textup{large}}}

\makeatletter
\newcounter{algocounter}
\@ifpackageloaded{hyperref}%
  {\newcommand{\mylabel}[2]
    {\refstepcounter{algocounter}\protected@write\@auxout{}{\string\newlabel{#1}{{\textcolor{black}{\textup{#2}}}{\thepage}%
      {\@currentlabelname}{\@currentHref}{}}}}}%
\makeatother

\newcommand{\apmc}{{\sc apmc}\xspace}

\newcommand{\tO}{\tilde{O}}

\begin{document}

\title{A Nearly Optimal All-Pairs Min-Cuts Algorithm in Simple Graphs}
\author{Jason Li\\Carnegie Mellon University \and Debmalya Panigrahi\\Duke University \and Thatchaphol Saranurak\\University of Michigan}
\date{}

\maketitle

\eat{
\begin{abstract}
    The Gomory-Hu tree or cut equivalent tree (Gomory and Hu, 1961) is a classic 
    data structure for reporting $s-t$ min-cuts for all pairs of vertices $s$ and
    $t$ in an undirected graph. In a recent breakthrough, Abboud {\em et al.} 
    (FOCS 2021) obtained the first strictly subcubic (i.e., $o(n^3)$-time) algorithm 
    for constructing a Gomory-Hu tree of an unweighted graph on $n$ vertices. 
    Their algorithm runs in 
    $\tO(n^{2.5})$ time, and they posed the following main open question:
    is it possible to solve Gomory-Hu tree in quadratic (in $n$) time for simple graphs?
    Clearly, this would be 
    nearly optimal since the number of edges can be $\Omega(n^2)$. We answer this question in the affirmative in this paper by giving an algorithm that runs in $n^{2+o(1)}$ time. 
\end{abstract}
}
\begin{abstract}
We give an $n^{2+o(1)}$-time algorithm for finding $s$-$t$ min-cuts
for all pairs of vertices $s$ and $t$ in a simple, undirected graph on $n$ vertices. 
We do so by constructing a
Gomory-Hu tree (or cut equivalent tree) in the same running time,
thereby improving on the recent bound of $\Otil(n^{2.5})$
by Abboud {\em et al.} (STOC 2021). Our running time is nearly optimal as a function of $n$.
%
\end{abstract}
\pagenumbering{gobble}

\clearpage

\pagenumbering{arabic}

\section{Introduction}


An $s$-$t$ mincut is a minimum (weight/cardinality) set of edges in a graph whose removal disconnects two vertices $s, t$. Finding $s$-$t$ mincuts, and by duality the value of $s$-$t$ maxflows, is a foundational question in graph algorithms. Na\"ively, mincuts for all vertex pairs can be computed by running a maxflow algorithm separately for each vertex pair, thereby incurring $\Theta(n^2)$ maxflow calls on an $n$-vertex graph. In 1961, Gomory and Hu~\cite{GomoryH61} gave a remarkable result where they constructed a cut equivalent tree (or Gomory-Hu tree, after its inventors) that captures an $s$-$t$ mincut for every vertex pair $s,t$ using just $n-1$ maxflow calls. By plugging in the current fastest maxflow algorithm~\cite{BrandLLSSSW21minimum}, this gives an $\tO(mn + n^{5/2})$-time\footnote{$\tO(\cdot)$ suppresses poly-logarithmic factors.} algorithm for the all pairs min-cuts (\apmc) problem on an $n$-vertex, $m$-edge graph. Improving on Gomory and Hu's 60-year old algorithm for the \apmc problem on general, weighted graphs remains a major open question in graph algorithms.

For {\em unweighted} graphs however, we can do better. The first algorithm to do so was by Bhalgat~{\em et al.}~\cite{HariharanKPB07}, who used Steiner mincuts to obtain a running time of $\tO(mn)$ in unweighted graphs. Karger and  Levine~\cite{KargerL15} matched this bound using the same counting technique, but by a different algorithm based on randomized maxflow computations. In simple graphs, both these algorithms obtain a running time of $\tO(n^3)$ since $m = O(n^2)$. The first subcubic (in $n$) running time was recently obtained in a beautiful work by Abboud~{\em et al.}~\cite{AbboudKT20subcubic}, who achieved a running time of $\tO(n^{2.5})$ for simple graphs. 
They write: ``Perhaps the most interesting open question is whether $\Otil(m)$ time can be achieved, even in simple graphs and even assuming a linear-time maxflow algorithm.'' Interestingly, they also isolate why breaking the $n^{2.5}$ bound is challenging, and say: ``\ldots perhaps it will
lead to the first conditional lower bound for computing a Gomory-Hu tree.''


In this paper, we give an $n^{2+o(1)}$-time Gomory-Hu tree algorithm in simple graphs, thereby improving on the $\Otil(n^{2.5})$ bound of Abboud {\em et al.} 
Our result is {\em unconditional} -- specifically, we do not need to assume an $\tO(m)$-time maxflow algorithm. As a consequence, we also refute the possibility of a $n^{2.5}$ lower bound for the Gomory-Hu tree problem. 
Since there are ${n\choose 2} = \Theta(n^2)$ vertex pairs, the running time of our algorithm is near-optimal for the all-pair mincuts problem. 
Even if one were to only construct a Gomory-Hu tree (and not report the mincut values explicitly for all vertex pairs), our algorithm is near-optimal as a function of $n$ since $m$ can be $\Theta(n^2)$. 

Our main theorem is the following:

\begin{theorem}
\label{thm:main}There is an algorithm $\GHtree(G)$ that, given a
simple $n$-vertex $m$-edge graph $G$, with high probability computes
a Gomory-Hu tree of $G$ in $n^{2+o(1)}$ time. 
\end{theorem}



\eat{
\begin{itemize}
\item Quote \cite{AbboudKT20subcubic}  ``Perhaps the most interesting open question is whether $\Otil(n^{2})$ time can be achieved, even in simple graphs and even assuming Hypothesis 1.5 (near-linear time max flow)''
\item Our algorithm significantly improves the $\Otil(n^{11/4})$-time algorithm by Abboud et al.~\cite{AbboudKT20subcubic}. Even after
assuming near-linear maxflow algorithm, their algorithm still takes
$\Otil(n^{2.5})$ time.
\item As an immediate consequence, we near-optimally give an algorithm for
computing all-pairs min cuts because of the output size is already
$\Omega(n^{2})$.
\end{itemize}

\begin{cor}
Near optimal all-pairs min cuts algorithm. 
\end{cor}
}

Our techniques also yield a faster Gomory-Hu tree algorithm in sparse graphs.
The previous record for sparse graphs is due to another recent algorithm of 
Abboud {\em et al.}~\cite{AbboudKT20soda} that 
takes $O(mc+\sum_{i=1}^{m/c}T(m,n,F_{i}))$ time, where $\sum_{i}F_{i}=O(m)$
and $T(m,n,F_{i})$ is the time complexity for computing a maxflow 
of value at most $F_{i}$. (Here, $c$ is a parameter that can be chosen 
by the algorithm designer to optimize the bound.) 
We improve this bound in the following theorem to 
$\Otil(mc)+\frac{n^{1+o(1)}}{c}\cdot T(m,n)$ where $T(m,n)$ is the time complexity 
for computing a maxflow. For comparison, if we 
assume an $\tO(m)$-time maxflow algorithm, then the running time
improves from $\tO(m^{1.5})$ in \cite{AbboudKT20soda} to $mn^{0.5+o(1)}$ 
in this paper. Using existing maxflow algorithms~\cite{KathuriaLS20,BrandLLSSSW21minimum}, 
the bound is 
$\tO(m\cdot g(m,n))$ in \cite{AbboudKT20soda} where $g(m,n) = \min(m^{1/2}n^{1/6}, m^{1/2} + n^{3/4})$,
and improves to $\sqrt{mn}\cdot n^{o(1)}\cdot g(m,n)$ in this paper.

\begin{theorem}
\label{thm:main sparse}There is an algorithm $\GHtreeSparse(G)$
that, given a simple $n$-vertex $m$-edge graph $G$, with high probability
computes a Gomory-Hu tree of $G$ in $\Otil(mc)+\frac{n^{1+o(1)}}{c}\cdot T(m,n))$
time where $T(m,n)$ denotes the time complexity for computing a maximum
flow on an $n$-vertex $m$-edge graph and $c$ is a parameter that
we can choose. 
\end{theorem}

Before closing this section, we mention some other results on Gomory-Hu 
trees, and consequently for the \apmc problem. 
Gusfield~\cite{Gusfield90} gave an algorithm that simplifies
Gomory and Hu's algorithm, particularly from an implementation perspective,
although it did not achieve an asymptotic improvement in the running time.
If one allows a $(1+\epsilon)$ approximation, then faster algorithms 
are known; in fact, the problem can be solved using (effectively) $\polylog(n)$ 
maxflow calls~\cite{AbboudKT20focs,LiP21approximate}. Finally, 
there is a robust literature on Gomory-Hu 
tree algorithms for special graph classes. This includes near-linear 
time algorithms for the class of planar graphs~\cite{BorradaileSW10} and more 
generally, for surface-embedded graphs~\cite{BorradaileENW16}, 
as well as improved runtimes for graphs with bounded treewidth~\cite{ArikatiCZ98, AbboudKT20focs}.
For more discussion on the problem, the reader is referred to a survey
in the Encyclopedia of Algorithms~\cite{Panigrahi16}.



\eat{
\section{Overview}
\begin{itemize}
\item Technical contribution: well-linked set instead of expanders. 
\item Notation: Gomory-Hu tree or cut-equivalent tree?
\end{itemize}
}

\paragraph{Organization.}

In \Cref{sec:algo}, we introduce the tools that we need for 
our Gomory-Hu tree algorithm. We then give the 
Gomory-Hu tree algorithm using these tools, and prove 
\Cref{thm:main} and \Cref{thm:main sparse}.
In subsequent sections, we show how to implement each 
individual tool and establish their respective properties. %

\section{Gomory-Hu Tree Algorithm}
\label{sec:algo}

We start this section by defining Gomory-Hu trees. It will be convenient 
to also define {\em partial} Gomory-Hu trees which will play an 
important role in our algorithm.

\begin{defn}
[Partial Gomory-Hu trees] Let $G=(V,E)$ be a graph. A \emph{partial Gomory-Hu
tree} or simply a \emph{partial tree} $(T,\cP)$ of $G$ satisfies the following:
\begin{itemize}
\item $T$ is a tree on $V(T)\subseteq V$ called a \emph{terminal} set, 
\item $\cP$ is a partition of $V$ where each part $V_{u}\in\cP$ contains
exactly one terminal $u$,
\item for any pair of terminals $u,v\in V(T)$, a $u$-$v$ mincut $(A_{T},B_{T})$
in $T$ corresponds to a $u$-$v$ mincut $(A,B)$ in $G$ where $A=\bigcup_{x\in A_{T}}V_{x}$
and $B=\bigcup_{y\in B_{T}}V_{y}$.
\end{itemize}
If $V(T)=V$, then $T$ is a \emph{Gomory-Hu tree} of $G$.
\end{defn}

\paragraph{Terminology about Partial Trees.}
Let $X\subseteq V$ be a vertex set. We say that a partial tree $(T,\cP)$ \emph{captures all mincuts separating $X$ of size at most $d$} if, for every part $U\in\cP$ and every pair of vertices $u,v\in U\cap X$, $\mincut_{G}(u,v)>d$. When $X=V$, we say that $(T,\cP)$ \emph{captures all mincuts of size at most $d$.} If all edges of $T$ have weight at most $d$, then we say that $(T,\cP)$ \emph{captures no mincut of size more than $d$.}

We say that $(T',\cP')$
is a \emph{refinement} of $(T,\cP)$ if $(T,\cP)$ can be obtained from
$(T',\cP')$ by contracting subtrees of $T'$ and taking the union of the corresponding
parts of $\cP'$. In other words, a refinement adds edges while preserving
the properties of a partial tree. The classic algorithm of Gomory and 
Hu~\cite{GomoryH61} starts with a vacuous partial tree comprising a single 
node and refines it in a series of $n-1$ iterations, where each iteration
adds a single edge to the tree. Our goal is to refine the partial tree 
faster by adding multiple edges in a single iteration. 
\eat{
We say that a partial tree $(T,\cP)$ \emph{captures 
all cuts of size at most $d$} if, for every part $U\in\cP$ and every 
pair of vertices $u,v\in U$, $\mincut_{G}(u,v) > d$. 
In a single iteration, we refine a partial tree that captures all 
cuts of size at most $d$ to one that captures all cuts of size at most
$2d$. Iterating over $d = 1, 2, 4, \ldots, n/2, n$ then gives our overall 
algorithm using only $O(\log n)$ iterations.
}
\eat{
We say that a partial tree $(T,\cP)$ \emph{captures 
all cuts of size at least $d$} if, for every part $U\in\cP$ and every 
pair of vertices $u,v\in U$, $\mincut_{G}(u,v)<d$. 
In a single iteration, we refine a partial tree that captures all 
cuts of size at least $2d$ to one that captures all cuts of size at least
$d$. Iterating over $d = n, n/2, n/4, \ldots, 1$ then gives our overall 
algorithm using only $O(\log n)$ iterations.
}

\paragraph{Well-linked Decomposition.}
The key to defining a single iteration of our algorithm that refines a 
partial tree is the notion of a well-linked decomposition. We first define 
a \emph{well-linked} set of vertices. 
\begin{defn}
We say that a vertex set $X$ is $(d,\phi)$-well-linked in a graph
$G$ if 
\begin{itemize}
\item For each $v\in X$, $\deg_{G}(v)\ge d$, where $\deg_{G}(v)$ is the degree of vertex $v$ in graph $G$, and
\item For each partition $(A,B)$ of $X$, $\frac{\mincut_{G}(A,B)}{d\cdot \min\{|A|,|B|\}}\ge \phi$. Here, $\mincut_G(A, B)$ is the smallest cut of $G$ that has vertex subsets $A$ and $B$ on different sides of the cut.
\end{itemize}\end{defn}
The next lemma is an important technical contribution of our paper,
and says that the set of high-degree vertices can be partitioned
into a small number of well-linked sets. 
Actually, this is the only place in this
paper where we require that the input graph $G$ is a simple graph.

\begin{lem}
\label{lem:partition}
There is an algorithm $\Partition(G,d)$ that,
given a simple $n$-vertex $m$-edge graph $G$ and a parameter $d$,
returns with high probability a partition $\{X_{1},\dots,X_{k}\}$
of $V_{\ge d}=\{v\mid\deg_{G}(v)\ge d\}$ such that $k=\Otil(n/d)$
and every set $X_{i}$ is $(d,\phipart)$-well-linked in $G$, where
$\phipart=n^{-o(1)}$. The algorithm $\Partition(G,d)$ 
runs in $m^{1+o(1)}$ time.
\end{lem}

In a single iteration, our goal is to refine a partial tree that captures 
mincuts of size at most $d$ to one that captures mincuts of size 
at most $2d$.
For this, we would like to partition all the vertices 
in $V_{\ge d}$ using the above lemma, and repeatedly refine 
the partial tree so that it captures all mincuts of size at most $2d$ separating the terminal set that includes the vertices in the $(d, \phi)$-well-linked set $X_i$. 
But, doing this on the input graph $G$ would be too slow; instead,
we use a sparse connectivity certificate that preserves all 
cuts of size at most $3d$. This suffices since in this iteration, 
we only seek to capture cuts of size at most $2d$. 
\eat{
For the latter, it is sufficient to ensure that 
all vertices of degree at least $d$ are terminals in the partial
tree. Since Lemma~\ref{lem:partition} partitions these vertices into 
a small number of $(d, \phi)$-well-linked sets, it then suffices to give 
an algorithm that refines a partial tree by adding vertices in 
a $(d, \phi)$-well-linked set as terminals. Indeed, this single 
step of refining a partial tree with a well-linked set 
will constitute our main lemma. 
}
\paragraph{Connectivity Certificate.} 
We formally define connectivity certificates next.
\begin{defn}
For any graph $G=(V,E)$, a \emph{$k$-connectivity certificate} $H$
of $G$ is a subgraph of $G$ that preserves all cuts in $G$ of size
$< k$, and ensures that all cuts in $G$ of size 
$\ge k$ have size $\ge k$ in $H$ as well. In other words, for any cut $(S,V\setminus S)$,
we have $|E_{H}(S,V\setminus S)| \ge \min\{|E_{G}(S,V\setminus S)|, k\}$.
\end{defn}

\eat{
In a single iteration of our algorithm, the input is a partial tree 
that already captures all cuts of size at least $2d$. Consequently, 
the refinement algorithm can operate on a 
$2d$-connectivity certificate of the input graph, instead of the 
entire graph. 
}
The next lemma, due to Nagamochi and 
Ibaraki~\cite{NagamochiI92}, gives an efficient algorithm for 
obtaining a connectivity certificate.

\begin{lem}
[\cite{NagamochiI92}]\label{lem:sparsify}There is an algorithm $\Sparsify(G,k)$
that, given an $n$-vertex $m$-edge graph $G$ and a parameter
$k$, return a $k$-connectivity certificate $H$ of $G$ with at
most $\min\{m,nk\}$ edges in $O(m)$ time.
\end{lem}

\paragraph{The Main Lemma.}
We are now ready to state our main lemma, which constitutes a 
refinement of the partial tree.


\eat{
\begin{lem}
\label{lem:refineOnSparsifier}There is an algorithm $\RefineOnSparsifier(G,H,(T,\cP),X,\phi)$
that, given 
\begin{itemize}
\item graphs $G$ and $H$ on the same $n$ vertices but with $m$ and $m'$ edges respectively,
where $H$ is a $2d$-connectivity certificate of $G$ for some value $d$, 
\item a partial tree $(T,\cP)$ of $G$ that captures all cuts of size at
least $2d$, and 
\item a set $X$ that is $(d,\phi)$-well-linked in $H$, 
\end{itemize}
returns with high probability a partial tree $(T',\cP')$ of $G$
that is a refinement of $(T,\cP)$ where $V(T')=V(T)\cup X$ in $O(m')$
time plus $\tO(1/\phi)$ max flow calls where each graph instance
has $O(n)$ vertices and $O(m')$ edges.
\end{lem}
}

\begin{lem}
\label{lem:refineOnSparsifier}
There is an algorithm $\RefineOnSparsifier(G,H,(T,\cP),X,d,\phi)$
that given
\begin{itemize}
\item graph $G$ on $n$ vertices and $m$ edges, and a $3d$-connectivity certificate
$H$ of $G$ with $m'\le \min\{m,3nd\}$ edges,  
\item a partial tree $(T,\cP)$ of $G$ that captures all mincuts of size
at most $d$ and no mincut of size more than $2d$, and 
\item a set $X$ that is $(d,\phi)$-well-linked in $H$,
\end{itemize}
in $\Otil(\min\{m,nd\} / \phi)$ time plus max-flow calls on
several graph instances with $\Otil(n/\phi)$ vertices and $\Otil(\min\{m,nd\}/\phi)$ edges in total, 
returns with high probability a partial tree $(T',\cP')$ of $G$ where 
\begin{itemize}
\item $(T',\cP')$ is a refinement of $(T,\cP)$, and
\item $(T',\cP')$ captures all mincuts separating $X\cup V(T)$ of size
at most $2d$ and no mincut of size more than $2d$.
\end{itemize}
\end{lem}

Crucially, when $3nd \le m$, the running time in the above lemma does not depend on $m$, 
the number of edges in $G$. In other words, the algorithm does not
even read in the entire graph $G$, instead operating on the $3d$-connectivity
certificate $H$ directly.


\paragraph{Small Connectivities.}
Recall that in a single iteration, \Cref{lem:partition} produces
$\tO(n/d)$ sets each of which is $(d, \phipart)$-well-linked, and \Cref{lem:refineOnSparsifier} 
makes max-flow calls on graphs with $\Otil(n/\phipart)$ vertices and $\Otil(nd/\phipart)$ edges in total. 
The current fastest max flow algorithm gives the following
runtime:
\begin{theorem}[\cite{BrandLLSSSW21minimum}]\label{thm:maxflow subroutine}There
is an algorithm that can find, with high probability, a maximum flow
on a graph with $n$ vertices and $m$ edges in $\Otil(m+n^{1.5})$
time.
\end{theorem}
Using this algorithm, the runtime of the max flow calls in 
an iteration becomes 
$\tO(\frac{n}{d})\cdot \tO(nd + n^{1.5}) \cdot n^{o(1)} = (n^{2} + \frac{n^{2.5}}{d})\cdot n^{o(1)}$ 
(recall that $\phipart = n^{-o(1)}$ in \Cref{lem:partition}). While this suffices for 
$d \ge \sqrt{n}$, we need an additional trick to handle small 
connectivities, namely $d < \sqrt{n}$. 

\eat{
We say that a partial tree $(T,\cP)$ \emph{captures
all cuts of size at most $d$} if, for every part $U\in\cP$ and every
pair of vertices $u,v\in U$, $\mincut_{G}(u,v)>d$. 
}
The next theorem, due to Hariharan~{\em et al.}~\cite{HariharanKP07}
and Bhalgat~{\em et al.}~\cite{HariharanKPB07}, gives a fast algorithm 
for computing a partial tree that captures all small cuts:
\begin{theorem}
[\cite{HariharanKP07,HariharanKPB07}]\label{thm:partial tree subroutine}There
is an algorithm $\PartialTree(G,d)$ that, given a simple $n$-vertex
$m$-edge graph $G$ and a parameter $d$, returns with high probability
a partial tree $(T,\cP)$ that captures all cuts of size at most $d$
in $\Otil(\min\{md, m+nd^2\})$ time. 
\end{theorem}
If we set $d = \sqrt{n}$, then this theorem gives a partial tree
that captures all cuts of size at most $\sqrt{n}$ in $\tO(n^2)$
time. We initialize our algorithm with this partial tree, and then
run the iterative refinement process described above for 
$d = \sqrt{n}, 2 \sqrt{n}, \ldots, n/2, n$ to obtain the Gomory-Hu tree.
We formally describe this algorithm below and prove its correctness
and runtime bounds.

\begin{algorithm}[H]
\begin{enumerate}
\item Initialize $(T,\cP)\gets\PartialTree(G,c)$ where $c$
is a parameter we can choose. 
\item For $d=c,2c,\ldots,n/2, n$
\begin{enumerate}

\item $H\gets\Sparsify(G,3d)$
\item $\{X_{1},\dots,X_{\Otil(n/d)}\}\gets\Partition(H,d)$
\item For each $X_{i}$, $(T,\cP)\gets\RefineOnSparsifier(G,H,(T,\cP),X_{i},d,\phipart)$
\end{enumerate}
\item Return $T$
\end{enumerate}
\caption{\label{alg:n2 time}$\protect\GHtree(G)$}
\end{algorithm}

\paragraph{The Gomory-Hu Tree Algorithm.} The algorithm is 
given in \Cref{alg:n2 time}. We first establish correctness
of the algorithm. The next property formalizes the progress
made by the algorithm in a single iteration of the 
for loop.

\begin{lem}
\label{lem:correct warm up}At the beginning of each for-loop iteration
of \Cref{alg:n2 time}, if $(T,\cP)$ is a partial tree of $G$ that captures all mincuts
of size at most $d$ and no mincut of size more than $d$, 
then at the end of the iteration, $(T,\cP)$
captures all mincuts of size at most $2d$ and no mincut of size more
than $2d$. 
\end{lem}

\begin{proof}
First, observe that the input to $\RefineOnSparsifier(\cdot)$ is
valid: (1) $H$ is a $3d$-connectivity certificate of $G$ containing
$\le \min\{m,3nd\}$ edges by \Cref{lem:sparsify}, (2) $(T,\cP)$ is a partial tree of $G$ that
captures all mincuts of size at most $d$ and no mincut of 
size more than $2d$ by assumption, and (3) $X$
is $(d,\phipart)$-well-linked in $H$ by \Cref{lem:partition}. 

By the second property in \Cref{lem:refineOnSparsifier}, 
$(T,\cP)$ captures no mincut of size more than $2d$. It remains
to show that at the end of the iteration, $(T,\cP)$ captures all 
mincuts of size at most $2d$. For mincuts of size at most $d$,
this follows from the assumption. Consider an $s$-$t$ mincut 
of size more than $d$ but at most $2d$. Since $s, t\in V_{\ge d}$
in $G$, it follows that $s, t\in V_{\ge d}$ in $H$ as well.
Thus, $s,t \in \cup_i X_i$ produced by $\Partition(H, d)$.
There are two cases. If $s, t\in X_i$ for some $i$,
then \Cref{lem:refineOnSparsifier} ensures that the 
$s$-$t$ mincut is captured by $(T,\cP)$ after 
$\RefineOnSparsifier(G,H,(T,\cP),X_{i},d,\phipart)$.
If $s\in X_i$, $t\in X_j$ where $i < j$ (wlog), then, 
when we call $\RefineOnSparsifier(G,H,(T,\cP),X_{j},d,\phipart)$, we have $s\in V(T)$ and $t\in X_j$.
Again, by \Cref{lem:refineOnSparsifier}, the
$s$-$t$ mincut is captured by $(T,\cP)$ after the call to 
$\RefineOnSparsifier$.
\end{proof}

The following is a simple corollary of the above lemma.
\begin{lem}
\label{lem:correct}\Cref{alg:n2 time} computes a Gomory-Hu tree $T$. 
\end{lem}

\begin{proof}
First, note that $(T,\cP)\gets\PartialTree(G,c)$ captures
all mincuts in $G$ of size at most $c$ by \Cref{thm:partial tree subroutine}.
Therefore, by \Cref{lem:correct warm up}, at the end of each iteration
of the for loop, \Cref{alg:n2 time} captures all mincuts of size
at most $2d$. As a consequence, at the end of the final loop, 
\Cref{alg:n2 time} captures all mincuts of size at most $n$.
Therefore, $T$ is indeed a Gomory-Hu tree.
\end{proof}

\eat{
Next, we make the following simple observations.
\begin{prop}
\label{prop:obs partial tree}We have the following:
\begin{enumerate}
\item \label{enu:all cuts}If $(T,\cP)$ captures all cuts of size
at most $d$ \emph{and} all cuts of size at least $d$, then $T$
is a Gomory-Hu tree. 
\item \label{enu:preserve}If $(T,\cP)$ captures all cuts of size at least
(at most) $d$, then its refinement $(T',\cP')$ also captures all
cuts of size at least (at most) $d$.
\end{enumerate}
\end{prop}

\begin{lem}
\label{lem:correct}\Cref{alg:n2 time} computes a Gomory-Hu tree $T$. 
\end{lem}

\begin{proof}
First, note that $(T,\cP)\gets\PartialTree(G,c)$ captures
all cuts in $G$ of size at most $c$ by \Cref{thm:partial tree subroutine}.
Moreover, before the first for-loop iteration, $(T,\cP)$ vacuously captures
all cuts of size at least $2n$ because there is no such cut. By
\Cref{lem:correct warm up}, after the last iteration, $(T,\cP)$ captures
all cuts of size at least $c$. As $(T,\cP)$ captured all cuts of
size at most $c$ even before the first iterations and remains so
after refinement by \Cref{prop:obs partial tree} (\Cref{enu:all cuts}).
Therefore, $T$ is indeed a Gomory-Hu tree by \Cref{prop:obs partial tree} (\Cref{enu:preserve}). 
\end{proof}
}

We now establish the running time of \Cref{alg:n2 time}.

\begin{lem}
\label{lem:time dense}By choosing $c=\sqrt{n}$, \Cref{alg:n2 time}
takes $n^{2+o(1)}$ time.
\end{lem}

\begin{proof}
$\PartialTree(G,c)$ takes $\Otil(m+nc^{2})=\Otil(n^{2})$ time 
by \Cref{thm:partial tree subroutine}. For each of the $O(\log n)$ iterations,
$\Sparsify(G,3d)$ takes $O(m)$ time (by \Cref{lem:sparsify})
and $\Partition(G,d)$ takes $m^{1+o(1)}=n^{2+o(1)}$ time
(by \Cref{lem:partition}). Since $H$ has
$O(nd)$ edges and $X_{i}$ is $(d,\phipart)$-well-linked, $\RefineOnSparsifier(G,H,(T,\cP),X_{i},d,\phipart)$
takes $(\min\{m,nd\}+n^{1.5})\cdot n^{o(1)} \le (nd+n^{1.5})\cdot n^{o(1)}$ time by \Cref{lem:refineOnSparsifier} 
and \Cref{thm:maxflow subroutine}.\footnote{
Note that since the running time is convex and each graph has at most
$nd$ edges and $n$ vertices, the worst case is when there 
are $n^{o(1)}$ maxflow calls on graphs with $nd$ edges and $n$ vertices.}
Since there are at most $\Otil(n/d)$ well-linked sets $X_{i}$, the
total time spent on $\RefineOnSparsifier$ is $(n^{2}+n^{2.5}/d)\cdot n^{o(1)}=n^{2+o(1)}$
since $d\ge c=\sqrt{n}$. The lemma follows by summing the time over all iterations.
\end{proof}
By analyzing the time differently, we obtain the following.
\begin{lem}
\label{lem:time sparse}For any parameter $c$, \Cref{alg:n2 time}
takes $\Otil(mc)+\frac{n^{1+o(1)}}{c}\cdot T(m,n)$ time where $T(m,n)$
denotes the time complexity for computing a maximum flow on an $n$-vertex
$m$-edge graph.
\end{lem}

\begin{proof}
$\PartialTree(G,c)$ takes
$\Otil(mc)$ time. For each of the $O(\log n)$ iterations, 
$\Sparsify(G,3d)$ takes $O(m)$ time (by \Cref{lem:sparsify}) and
$\Partition(G,d)$ takes $m^{1+o(1)} = m\cdot n^{o(1)}$ time (by \Cref{lem:partition}). 
Also, $\RefineOnSparsifier(\cdot)$ takes 
$(\min\{m, nd\}+T(m,n))\cdot n^{o(1)} \le (m+T(m,n))\cdot n^{o(1)} \le T(m,n)\cdot n^{o(1)}$ 
time by \Cref{lem:refineOnSparsifier}.\footnote{Note that 
since the running time is convex and each graph has at most
$m$ edges and $n$ vertices, the worst case is when there 
are $n^{o(1)}$ maxflow calls on graphs with $m$ edges and $n$ vertices.}
Since there are at most
$\Otil(n/d)=\Otil(n/c)$ well-linked sets $X_{i}$, the total time
spent on $\RefineOnSparsifier$ is $\frac{n^{1+o(1)}}{c}\cdot T(m,n)$.
The lemma follows by summing the time over all iterations.
\end{proof}
To conclude, observe that \Cref{thm:main} follows from \Cref{lem:correct}
and \Cref{lem:time dense}. Similarly, \Cref{thm:main sparse} follows
from \Cref{lem:correct} and \Cref{lem:time sparse}.

\section{Refinement with Well-linked Set}

\label{sec: refine}

Our goal in this section is to prove the main lemma (\Cref{lem:refineOnSparsifier}).
Let us first recall the setting of the lemma. We have a graph $G=(V,E)$
with $n$ vertices and $m$ edges and a $3d$-connectivity certificate
$H$ of $G$ containing $m'\le \min\{m, 3nd\}$ edges. Let $(T,\cP)$ be a partial
tree of $G$ that captures all mincuts of size at most $d$ and no
mincut of size more than $2d$. Let $X$ be a $(d,\phi)$-well-linked
set in $H$. For each terminal $u_{i}\in V(T)$ and its corresponding
part $V_{i}\in\cP$, let $X_{i}=V_{i}\cap X$. 

Now, we define the \emph{sparsified auxiliary graph} $H_{i}$. For
each connected component $C$ in $T\setminus\{u_{i}\}$, let $V_{C}=\bigcup_{u\in V(C)}V_{u}$
where each $V_{u}\in\cP$. The graph $H_{i}$ is obtained from $H$
by contracting $V_{C}$ into one vertex $u_{C}$ for every component
$C$ in $T\setminus\{u_{i}\}$. Let $n'_{i}$ and $m'_{i}$ denote
the number of vertices and edges in $H_{i}$ respectively.
($H_i$ is unweighted but not necessarily a simple graph.) Below,
we bound the total size of $H_{i}$ over all $i$. The bound on $\sum_{i}m'_{i}$
crucially exploits the fact that the graph is unweighted. 
\begin{prop}
\label{lem:edgebound}$\sum_{u_{i}\in V(T)}n'_{i}\le3n$ and $\sum_{u_{i}\in V(T)}m'_{i}\le \min\{3m,5nd\}$. 
\end{prop}

\begin{proof}
Observe that $n'_{i}=|V_{i}|+\deg_{T}(u_{i})$. So $\sum_{u_{i}\in V(T)}n'_{i}=n+2|V(T)|\le3n$.
Next, we bound $\sum_{u_{i}\in V(T)}m'_{i}$. For
any vertex $x\in V$, let $\rep(x)\in V(T)$ be the unique terminal
such that $x\in V_{\rep(x)}$. Consider each edge $(x,y)\in E(H)$.
Let $P_{xy}=(\rep(x),\dots,\rep(y))\subseteq V(T)$ be the unique
path in $T$ between $\rep(x)$ and $\rep(y)$. (Possibly $x$ and
$y$ are in the same part of $\cP$ and so $\rep(x)=\rep(y)$.) The
crucial observation is that an edge $(x,y)$ appears in $H_{i}$ if and only if the
terminal $u_{i}$ is in $P_{xy}$ (otherwise, $x$ are $y$ are contracted
into one vertex in $H_i$). That is, the contribution of $(x,y)$ to $\sum_{u_{i}\in V(T)}m'_{i}$
is exactly $|V_{T}(P_{xy})|=1+|E_{T}(P_{xy})|$. Summing over all
edges $e\in E(H)$, this implies that 
\[
\sum_{u_{i}\in V(T)}m'_{i}\le|E(H)|+\sum_{(x,y)\in E(H)}|E_{T}(P_{xy})|.
\]
Recall that $|E(H)|=m'\le \min\{m,3nd\}$. The last important observation is that $\sum_{(x,y)\in E(H)}|E_{T}(P_{xy})|$
is exactly the total weight of edges in $T$. This is because each $(x,y)\in E(H)$ contributes exactly one unit of weight to each tree-edge in $E_T(P_{xy})$. 
The total weight of edges in $T$ is at most $\min\{2m,2nd\}$. To see this, observe that 
it is at most $(|V(T)|-1)\cdot2d\le2nd$,
because $T$ has no edge with weight more than $2d$. 
Also, it is at most $\sum_{u_i\in V(T)}\deg_G(u_i) \le 2m$ because each tree edge $(u_i,u_j) \in E(T)$ has weight $\mincut_G(u_i,u_j) \le \min\{\deg_G(u_i),\deg_G(u_j)\}$.
This implies
the bound $\sum_{u_{i}\in V(T)}m'_{i}\le \min\{3m,5nd\}$  as claimed.
\end{proof}
The key step for proving \Cref{lem:refineOnSparsifier} is captured
by the following lemma.
\begin{lem}
\label{lem:refine on H}Given $H_{i}$, $X_{i}$, and $(T,\cP)$,
there is an algorithm that takes $\Otil(m'_{i}/\phi)$ time and additionally makes max-flow calls
on several graphs with $\Otil(n'_{i}/\phi)$ vertices and $\Otil(m'_{i}/\phi)$
edges in total, and then returns a partial
tree $(T'_{i},\cP'_{i})$ \textbf{of $\boldsymbol{G}$} such that 
\begin{itemize}
\item $(T'_{i},\cP'_{i})$ is a refinement of $(T,\cP)$, and
\item $(T'_{i},\cP'_{i})$ captures all mincuts separating $X_{i}\cup V(T)$
of size at most $2d$ and no mincut of size more than $2d$.
\end{itemize}
\end{lem}

Before proving \Cref{lem:refine on H}, we show that it implies \Cref{lem:refineOnSparsifier}. (See \Cref{fig:refine} for illustration.)

\begin{figure}
\centering
\includegraphics[height=1.2\textwidth]{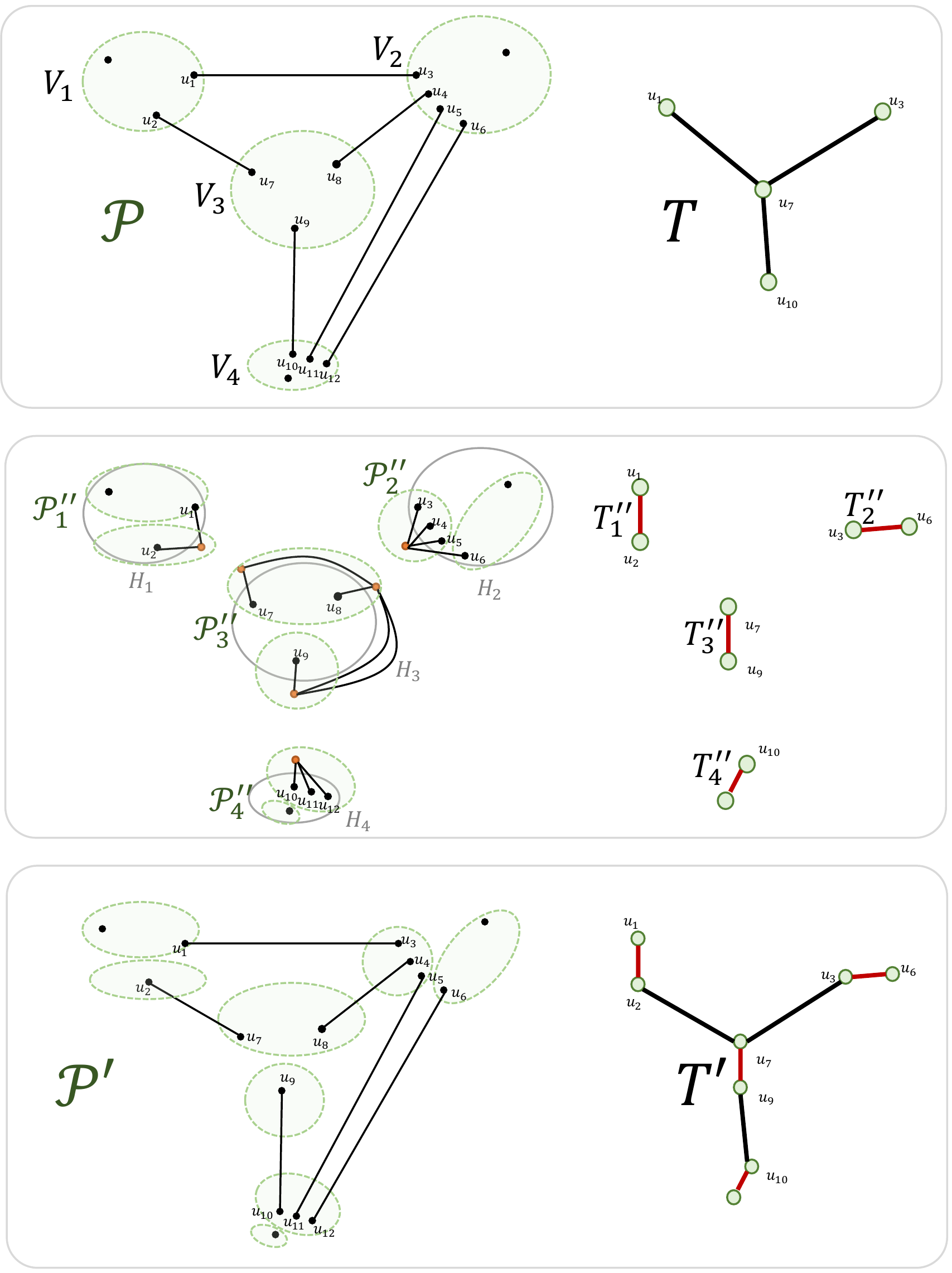}

\caption{Refining $(T,\cP)$ to $(T',\cP')$. The algorithm for \Cref{lem:refine on H} computes
partial trees $(T''_{i},\cP''_{i})$ of every sparsified auxiliary
graph $H_{i}$. This is illustrated in the second box in the figure
above. The algorithm in \Cref{lem:refine on H} computes a partial tree $(T'_{i},\cP'_{i})$
of $G$ which is a refinement of $(T',\cP)$ in (not shown in the
figure above). Then, in the proof of \Cref{lem:refineOnSparsifier}, we ``combine'' these refinement
on each part $V_{i}$ and we obtain the refined partial tree $(T',\cP')$
of $(T,\cP)$.\label{fig:refine}}

\end{figure}

\paragraph{Proof of \Cref{lem:refineOnSparsifier}. }

We apply \Cref{lem:refine on H} for all $i$ simultaneously
and obtain $(T'_{i},\cP'_{i})$ each of which refines $(T,\cP)$ in
exactly one part $V_{i}\in\cP$. Let $(T',\cP')$ be the refinement
of $(T,\cP)$ such that $(T',\cP')$ refines the part $V_{i}\in\cP$
according to $(T'_{i},\cP'_{i})$ for every $i$. 
Note that $(T',\cP')$
can be computed in $O(n)$ time. Clearly, $(T',\cP')$ captures no
mincut of size more than $2d$ (i.e., $T'$ has no edge of weight
more than $2d$) because none of $T'_{i}$ does. 

It remains to prove that $(T',\cP')$ captures all mincuts separating $V(T)\cup\Big(\bigcup_{i}X_{i}\Big)=V(T)\cup X$
of size at most $2d$. That is, there is no pair $x,y\in V(T)\cup X$
where $\mincut_{G}(x,y)\le2d$ and $x,y\in\cP'$ are in the same part.
This is true because, if $x$ and $y$ are from a different part of
$\cP$, then they are still from a different part in $\cP'$ as $\cP'$
is a refinement of $\cP$. Otherwise, if $x$ and $y$ are from the
same part of $\cP$, say $V_{i}\in\cP$, then $x,y\in X_{i}\cup V(T)$
and so \Cref{lem:refine on H} guarantees that they must be separated
by $\cP'_{i}$ and hence in $\cP'$. This concludes the correctness
of \Cref{lem:refineOnSparsifier}.

Next, we analyze the running time. The total running time is $\sum_{i}\Otil(m'_{i}/ \phi )=\Otil(\min\{m,nd\} / \phi)$
by \Cref{lem:edgebound} and \Cref{lem:refine on H}. 
Finally, the graphs that the algorithm makes max-flow calls on contain
in total at most $\sum_{i}\Otil(n'_{i}/\phi)=\Otil(n/\phi)$ vertices
and $\sum_{i}\Otil(m'_{i}/\phi)=\Otil(\min\{m,nd\}/\phi)$ edges by \Cref{lem:edgebound}.
This completes the proof.

\paragraph{Proof of \Cref{lem:refine on H}. }

For the remaining part of this section, we prove \Cref{lem:refine on H}.
There are two main ingredients. 

First, we show that the problem of creating a partial tree on a set
of terminals can be reduced to finding single source connectivity
on the terminals. This step closely mirrors \cite{LiP21approximate}:
while they focus on the \emph{approximate} Gomory-Hu tree problem, their
techniques translate over to the exact case. Nevertheless, since the reduction
might be of independent interest, we give the proof later in \Cref{sec:reduction}.
\begin{lem}
\label{lem:reduction}Let $G=(V,E)$ be an $n$-vertex $m$-edge unweighted (respectively, weighted) graph
with a terminal set $X\subseteq V$, and let $k\ge0$ be a real number. Suppose we have an oracle that,
given a terminal $p\in X$, returns $\min\{\mincut_{G}(p,v),k\}$
for all other terminals $v\in X$. Then, there is an algorithm that computes with high probability a partial tree $(T,\cP)$ of $G$ where $V(T) \subseteq X$ that captures
all mincuts separating $X$ of size at most $k$ and no mincuts of size more than $k$.
It makes calls to the oracle and max-flow on unweighted (respectively, weighted) graphs with a total of $\tilde{O}(n)$ vertices and $\tilde{O}(m)$ edges, and runs for $\tilde{O}(m)$ time outside of these calls.
\end{lem}

Note that it is crucial for us that the reduction above works even
when the oracle only returns $\min\{\mincut_{G}(p,v),k\}$ and not
$\mincut_{G}(p,v)$. The next lemma exactly implements this oracle:
\begin{lem}
\label{lem:single-source} Let $G=(V,E)$ be an $n$-vertex $m$-edge
graph. Let $X$ be a $(d,\phi)$-well-linked set in $G$. Let $p\in X$
be any fixed vertex in $X$. Then, there is an algorithm that computes
$\min\{\mincut_{G}(p,v),2d\}$ for all other $v\in X$ in
$O(\frac{m \log n}{\phi})$ time plus $\frac{\polylog(n)}{\phi}$ max-flow calls each on a graph with $O(n)$ vertices and $O(m)$ edges.
\end{lem}

We will prove \Cref{lem:single-source} in \Cref{sec:isolate_well-linked}.
First, we show how to apply both lemmas above to prove \Cref{lem:refine on H}.
We start with a simple observation.
\begin{prop}
$X_{i}$ is $(d,\phi)$-well-linked in $H_{i}$. 
\end{prop}

\begin{proof}
As $X$ is $(d,\phi)$-well-linked in $H$, any subset $X_{i}\subseteq X$
is also $(d,\phi)$-well-linked in $H$. Now, observe that the property that 
a vertex set is $(d,\phi)$-well-linked is preserved under graph contraction. 
As $H_{i}$ is a contracted graph of $H$, the proposition follows.
\end{proof}
By setting parameters $G\gets H_{i}$ and $X\gets X_{i}$ as the inputs
of \Cref{lem:single-source}, we obtain the required oracle for \Cref{lem:reduction}
when $k=2d$. By applying \Cref{lem:reduction}, we obtain a partial
tree $(T''_{i},\cP''_{i})$ of $H_{i}$ where $V(T''_i) \subseteq X$, that captures all mincuts
separating $X_{i}$ of size at most $2d$ and no mincuts of size more
than $2d$. 
This steps takes $\Otil(m'_{i}/\phi)$ time and makes max-flow calls
on several graphs with $\Otil(n'_{i}/\phi)$ vertices and $\Otil(m'_{i}/\phi)$
edges in total. 

We are not quite done as we need a partial tree of $G$ (not of $H_{i}$)
with all properties required by \Cref{lem:refine on H}, but the remaining
steps are quite easy. Suppose the vertex $u_{i}$, which was the unique
terminal in $V_{i}$ in the partial tree $(T,\cP)$, is now in part
$V_{u_{i}}\subseteq V(H)$ of the partition $\cP''_{i}$. Moreover,
let $x_{u_{i}}\in X_{i}\cap V(T''_{i})$ denote the unique terminal
of part $V_{u_{i}}\in\cP''_{i}$. The algorithm just checks if $\mincut_{H_{i}}(u_{i},x_{u_{i}})\le2d$
by using a single max-flow call on $H_{i}$. If so, we further refine
$(T''_{i},\cP''_{i})$ according to the mincut separating $u_{i}$ and $x_{u_{i}}$. If not,
then we let $u_{i}$ replace $x_{u_{i}}$ as a unique terminal of
part $V_{u_{i}}\in\cP''_{i}$. At this point, $(T''_{i},\cP''_{i})$
is a partial tree of $H_{i}$ that captures all mincuts separating
$X_{i}\cup\{u_{i}\}$ of size at most $2d$ and no mincuts of size
more than $2d$. Finally, we refine the part $V_{i}$ of $(T,\cP)$
according to $(T''_{i},\cP''_{i})$ and obtain a partial tree $(T'_{i},\cP'_{i})$
of $G$ as desired. 
The reason this is correct is because $(T''_{i},\cP''_{i})$
captures only mincuts of size at most $2d$ but $H$ preserves exactly
\emph{all} cuts of $G$ of size at most $3d$. The running times in
these final steps are subsumed by the previous steps. This completes the proof of \Cref{lem:refine on H}.

\subsection{Single-source Mincut Values for Well-linked Sets: Proof of \Cref{lem:single-source}}

\label{sec:isolate_well-linked} 

We recall the setting of \Cref{lem:single-source}. We have an $n$-vertex
$m$-edge graph $G$ and a $(d,\phi)$-well-linked set $X$ in $G$.
Let $p\in X$ be any fixed vertex in $X$. The goal is to compute
$\min\{\mincut_{G}(p,v),2d\}$ for all other $v\in X\setminus\{p\}$.

Now, we need to introduce some notation. We say that a cut $(A,B)$
in $G$ is an $(S,T)$-cut if $S\subseteq A$ and $T\subseteq B$.
Moreover, $(A,B)$ is an $(S,T)$-mincut if, additionally, $|E_{G}(A,B)|=\mincut_{G}(S,T)$.
We say that $(A,B)$ is \emph{the (unique) $S$-minimal} $(S,T)$-mincut if,
for any $(S,T)$-mincut $(A',B')$, we have $S\subseteq A\subseteq A'$.
We will not need the notion of minimal mincut in this section, but
it will be used later in \Cref{sec:reduction}. A key tool in proving
\Cref{lem:single-source} is the following \emph{Isolating Cuts Lemma} of Li and Panigrahi~\cite{LiP20deterministic}, which was discovered independently by Abboud, Krauthgamer, and Trabelsi~\cite{AbboudKT20subcubic}.%
\begin{lem}
[Isolating Cut Lemma \cite{LiP20deterministic,AbboudKT20subcubic}]\label{lem:isolating}
There is an algorithm that, given an undirected graph $G=(V,E)$ on
$n$ vertices and $m$ edges and a terminal set $T\subseteq V$, finds
the $t$-minimal $(t,T\setminus\{t\})$-mincut for every $t\in T$
in $O(m\log n)$ time plus $O(\log n)$ maxflow calls each on a graph with
$O(n)$ vertices and $O(m)$ edges. 
\end{lem}

Fix any $x\in X$ where $\mincut_G(p,x)\le 2d$. Let $(A,B)$ be any $(p,x)$-mincut where $p\in A$
and $x\in B$. We have three observations.
The first crucial observation says that $(A,B)$
must be ``unbalanced'' w.r.t.~$X$.
\begin{prop}
\label{lem:unbalanced}$\min(|A\cap X|,|B\cap X|)\le\frac{2}{\phi}$. 
\end{prop}

\begin{proof}
By the well-linkedness of $X$, we have $d\phi\cdot\min(|A\cap X|,|B\cap X|)\le|E(A,B)|$.
On the other hand, we have $|E(A,B)|\le\mincut_{G}(p,x)\le2d$. The bound follows by combining the two inequalities.
\end{proof}
Let $S$ be an i.i.d. sample of $X$ with rate $\phi/2$. Let $T=S\cup\{p\}$.
The second observation roughly says that, with probability
$\Omega(\phi)\ge1/n^{o(1)}$, one side of $(A,B)$ contains only one vertex from $T$.
\begin{prop}
\label{prop:prob iso}With probability at least $\phi/(2e)$, either 
\begin{itemize}
\item $A\cap T=\{p\}$ and $x\in T$, or 
\item $B\cap T=\{x\}$ and $p\in T$.
\end{itemize}
\end{prop}

\begin{proof}
By \Cref{lem:unbalanced}, either $|A\cap X|\le2/\phi$ or $|B\cap X|\le2/\phi$.
If $|A\cap X|\le2/\phi$, then we have 
\[
\Pr[A\cap T=\{p\}\text{ and }x\in T]=\Pr[(A\setminus\{p\})\cap S=\emptyset]\cdot\Pr[x\in S]=\left(1-\frac{\phi}{2}\right)^{|A\cap X|-1}\cdot\frac{\phi}{2}\ge\frac{1}{e}\cdot\frac{\phi}{2}
\]
If $|B\cap X|\le2/\phi$, then we have 
\[
\Pr[B\cap T=\{x\}\text{ and }p\in T]=\Pr[(B\setminus\{x\})\cap S=\emptyset]\cdot\Pr[x\in S]=\left(1-\frac{\phi}{2}\right)^{|B\cap X|-1}\cdot\frac{\phi}{2}\ge\frac{1}{e}\cdot\frac{\phi}{2}.
\]
\end{proof}
The last observation says that given that the event in \Cref{prop:prob iso}
happens, then either the $(p,T\setminus\{p\})$-mincut or the $(x,T\setminus\{x\})$-mincut
is a $(p,x)$-mincut. This will be useful for us because the Isolating
Cut Lemma can compute the $(p,T\setminus\{p\})$-mincut and the $(x,T\setminus\{x\})$-mincut
quickly. 
\begin{prop}
\label{prop:good cut}We have the following:
\begin{enumerate}
\item If $A\cap T=\{p\}$ and $x\in T$, then any $(p,T\setminus\{p\})$-mincut
is a $(p,x)$-mincut.
\item If $B\cap T=\{x\}$ and $p\in T$, then any $(x,T\setminus\{x\})$-mincut
is a $(p,x)$-mincut.
\end{enumerate}
\end{prop}

\begin{proof}
(1): As $A\cap T=\{p\}$, $(A,B)$ is a $(p,T\setminus\{p\})$-cut
and so $\mincut(p,T\setminus\{p\})\le|E(A,B)|=\mincut(p,x)$. Since
$x\in T$, any $(p,T\setminus\{p\})$-cut is a $(p,x)$-cut. Therefore,
a $(p,T\setminus\{p\})$-mincut is a $(p,x)$-cut of size at most
$\mincut(p,x)$. So it is a $(p,x)$-mincut.

(2): The proof is symmetric. As $B\cap T=\{x\}$, $(B,A)$ is a $(x,T\setminus\{x\})$-cut
and so $\mincut(x,T\setminus\{x\})\le|E(A,B)|=\mincut(p,x)$. Since
$p\in T$, any $(x,T\setminus\{x\})$-cut is a $(p,x)$-cut. Therefore,
a $(x,T\setminus\{x\})$-mincut is a $(p,x)$-cut of size at most
$\mincut(p,x)$. So it is a $(p,x)$-mincut.
\end{proof}
\begin{algorithm}
\begin{enumerate}
\item Initialize $\mathtt{val}[x]=2d$ for all $x\in X\setminus\{p\}$.
\item Repeat $c\cdot\frac{\ln n}{\phi}$ times (for a large enough constant
$c$) 
\begin{enumerate}
\item Sample $S$ from $X$ i.i.d. at rate $\frac{\phi}{2}$.
\item Call the Isolating Cuts Lemma (\Cref{lem:isolating}) on terminal set
$T=S\cup\{p\}$ and obtain a $(t,T\setminus\{t\})$-mincut $C_{t}$
of size $\delta(C_{t})$ for every $t\in T$.
\item For each $x\in S\setminus\{p\}$, do the following:
\begin{enumerate}
\item If $C_{x}$ is a $(p,x)$-cut (i.e.,~$p\notin C_{x}$), then $\mathtt{val}[x]\gets\min\{\mathtt{val}[x],\delta(C_{x})\}$.
\item If $C_{p}$ is a $(p,x)$-cut (i.e.,~$x\notin C_{p}$), then $\mathtt{val}[x]\gets\min\{\mathtt{val}[x],\delta(C_{p})\}$.
\end{enumerate}
\end{enumerate}
\item Return $\mathtt{val}[\cdot]$.
\end{enumerate}
\caption{$\textsc{SingleSourceMincut}(G,X,d,\phi,p)$\label{alg:iso well-linked}}
\end{algorithm}

The above observations directly suggest an algorithm stated in \Cref{alg:iso well-linked}.
Below, we prove its correctness in \Cref{lem:iso correct} and bound
the running time in \Cref{lem:iso time}. 

\begin{lem}
\label{lem:iso correct}\Cref{alg:iso well-linked} computes, with
high probability, $\mathtt{val}[x]=\min\{2d,\mincut(p,x)\}$ for all
$x\in X\setminus\{p\}$.
\end{lem}

\begin{proof}
Note that $\mathtt{val}[x]\le2d$ from initialization. So we only
need to show that if $\mincut(p,x)\le2d$, then $\mathtt{val}[x]=\mincut(p,x)$
whp. On one hand, $\mathtt{val}[x]\ge\mincut(p,x)$ because whenever
$\mathtt{val}[x]$ is decreased, it is assigned the size of some
$(p,x)$-cut (which is either $C_{x}$ or $C_{p}$). On the other
hand, we claim $\mathtt{val}[x]\le\mincut(p,x)$ whp. To see this,
observe that, with probability at least $1-(1-\phi/(2e))^{c\cdot\frac{\ln n}{\phi}}\ge1-1/n^{10}$,
that there exists an iteration in \Cref{alg:iso well-linked} where
the event in \Cref{prop:prob iso} happens. That is, $A\cap T=\{p\}$
and $x\in T$, or $B\cap T=\{x\}$ and $p\in T$. Given this, by \Cref{prop:good cut},
either a $(x,T\setminus\{x\})$-mincut $C_{x}$ or a $(p,T\setminus\{p\})$-mincut
$C_{p}$ is a $(p,x)$-mincut and so the algorithm sets $\mathtt{val}[x]\le\mincut(p,x)$.
\end{proof}
\begin{lem}
\label{lem:iso time}\Cref{alg:iso well-linked} takes $O\left(\frac{m\log^2 n}{\phi}\right)$ time plus $O\left(\frac{\log^{2}n}{\phi}\right)$
max-flow calls each on a graph with $O(n)$ vertices and $O(m)$ edges. 
\end{lem}

\begin{proof}
Note that $O\left(\frac{\log n}{\phi}\right)$ invocations of \Cref{lem:isolating}
takes $O\left(\frac{\log^{2}n}{\phi}\right)$ max-flow calls each
on a graph with $O(n)$ vertices and $O(m)$ edges plus $O\left(\frac{m\log^2 n}{\phi}\right)$
time. Additionally, for each invocation of \Cref{lem:isolating},
we update $\mathtt{val}[\cdot]$ in $O(n)$ time for a total of
$O\left(\frac{n\log n}{\phi}\right)$
time. 
\end{proof}
By \Cref{lem:iso correct} and \Cref{lem:iso time}, this completes
the proof of \Cref{lem:single-source}.

\section{Well-linked Partitioning}
The goal of this section is to prove \Cref{lem:partition}.
We start with some notation.
For disjoint vertex subsets $V_1,\ldots,V_\ell\subseteq V$, define $E_G(V_1,\ldots,V_\ell)$ as the set of edges $(u,v)\in E$ with $u\in V_i$ and $v\in V_j$ for some $i\ne j$. For a vector $\mathbf{d}\in\mathbb R^V$ of entries on the vertices, define $\mathbf{d}(v)$ as the entry of $v$ in $\mathbf{d}$, and for a subset $U\subseteq V$, define $\mathbf{d}(U):=\sum_{v\in U}\mathbf{d}(v)$.
We now introduce the concept of an expander ``weighted'' by \emph{demands} on the vertices. 

\begin{definition}[$(\phi,\mathbf{d})$-expander]
Consider a weighted, undirected graph $G=(V,E)$ with edge weights $w$ and a vector $\mathbf{d}\in\mathbb R^V_{\ge0}$ of non-negative ``demands'' on the vertices. The graph $G$ is a \emph{$(\phi,\mathbf{d})$-expander} if for all subsets $S\subseteq V$,
\[ \frac{|E_G(S,V\setminus S)|}{\min\{\mathbf{d}(S),\mathbf{d}(V\setminus S)\}} \ge \phi.\]
\end{definition}

We now state the algorithm of~\cite{LiS21} that computes our desired expander decomposition, which generalizes the result from \cite{ChuzhoyGLNPS20}.


\begin{theorem}[$(\phi,\mathbf{d})$-expander decomposition algorithm~\cite{LiS21}]\label{thm:exp}
Fix any $\epsilon>0$ and any parameter $\phi>0$. Given a weighted, undirected graph $G=(V,E)$ with edge weights $w$ and a non-negative demand vector $\mathbf{d}\in\mathbb R^V_{\ge0}$ on the vertices, there is a deterministic algorithm running in $m^{1+\epsilon} (\lg n)^{O(1/\epsilon^2)}$
time that partitions $V$ into subsets $V_1,\ldots,V_\ell$ such that
 \begin{enumerate}
 \item For each $i\in[\ell]$, define the demands $\mathbf{d}_i\in\mathbb R^{V_i}_{\ge0}$ as $\mathbf{d}$ restricted to the vertices in $V_i$. Then, the graph $G[V_i]$ is a $(\phi,\mathbf{d}_i)$-expander.
 \item The total number $|E_G(V_1,\ldots,V_\ell)|$ of inter-cluster edges is $B \phi \cdot \mathbf{d}(V)$ where $B = (\lg n)^{O(1/\epsilon^4)}$.
 \end{enumerate}
\end{theorem}

Given \Cref{thm:exp}, we can apply it to obtain the desired well-linked sets using the following key lemma:

\begin{lemma}\label{lem:subset-U}
There is an algorithm that, given any subset $U\subseteq V_{\ge d}=\{v\mid \deg_G(v)\ge d\}$, outputs disjoint subsets $X_1,\ldots,X_k$ of $U$ such that $k \le 2n/d$, every set $X_i$ is $(d,\phi)$-well-linked in $G$ for $\phi=n^{-o(1)}$, and $|\cup_i X_i| \ge |U|/2$. This algorithm runs in $m^{1+o(1)}$ time.
\end{lemma}
\begin{proof}
Apply \Cref{thm:exp} with $\phi=\frac1{8B}$ (recall that $B = (\lg n)^{O(1/\epsilon^4)}$), and the following demands: $\mathbf{d}(v)=d$ for all $v\in U$ and $\mathbf{d}(v)=0$ for all $v\notin U$. (We will set the value of $\epsilon$ later.) We obtain a partition $V_1,\ldots,V_\ell$ of $V$ with $|E_G(V_1,\ldots,V_\ell)|\le B\phi\cdot\mathbf{d}(V) = \mathbf{d}(V)/8 = d\cdot |U|/8$. For each $i\in[\ell]$ and vertex $v\in U\cap V_i$, assign $v$ the value $x(v)=\frac{|E_G(V_i,V\setminus V_i)|}{|U\cap V_i|}$, so that $\sum_{v\in U}x(v)=2|E_G(V_1,\ldots,V_\ell)|\le d\cdot |U|/4$. If we select a vertex $v\in U$ uniformly at random, then the expected value of $x(v)$ is at most $d/4$; so, by Markov's inequality, we have $x(v)\le d/2$ with probability at least $1/2$. Let $U'\subseteq U$ be all vertices $v\in U$ with $x(v)\le d/2$; it follows that $|U'|\ge|U|/2$. For each subset $V_i$, the value of $x(v)$ is identical for all vertices in $U\cap V_i$. Hence, either $U\cap V_i$ is contained in $U'$ or is disjoint from it; without loss of generality, let $V_1,\ldots,V_k$ be the sets that are contained in $U'$ for some $k\le\ell$. We now set $X_i=U\cap V_i$ for all $i\in[k]$.

We first show that each set $X_i$ is $(d,\phi)$-well-linked. Since $X_i\subseteq U\subseteq V_{\ge d}$, we have $\deg_G(v)\ge d$ for all $v\in X_i$. Now consider a partition $(A,B)$ of $X_i$. For any subset $S\subseteq V_i$ that contains $A$ and is disjoint from $B$, we have $|E_G (S,V_i\setminus S)|\ge \phi \cdot \min\{\mathbf{d}(S),\mathbf{d}(V\setminus S)\} = d\phi\cdot \min\{|A|,|B|\}$, where the inequality holds by definition of $(\phi,\mathbf{d})$-expander. It follows that $\text{mincut}_G(A,B)\ge d\phi\cdot \min\{|A|,|B|\}$,
and hence, $X_i$ is $(d,\phi)$-well-linked.

We now show that $|V_i|\ge d/2$ for all $i\in[k]$; since the $V_i$ are disjoint, this would imply that $k\le 2n/d$. Recall that $X_i=U\cap V_i$, so that $|E_G(V_i,V\setminus V_i)|=\sum_{v\in X_i}x(v)\le |X_i|\cdot d/2$. By averaging, there exists $v\in X_i$ with $|E_G(v,V\setminus V_i)|\le d/2$. Since $\deg_G(v)\ge d$, at least $d/2$ edges incident to $v$ must have their other endpoint inside $V_i$. Since $G$ is simple, the endpoints must be distinct, so $|V_i|\ge d/2$, as promised.

Finally, we fix the value of $\epsilon = (\lg n)^{-1/5}$. Then, 
$$\phi = \frac{1}{8B} = \frac{1}{8 (\lg n)^{O(1/\epsilon^4)}} = \frac{1}{8 (\lg n)^{O((\lg n)^{4/5})}} = \frac{1}{n^{o(1)}}.$$
The running time is $m^{1+\epsilon}\cdot (\lg n)^{O(1/\epsilon^2)} = m^{1+(\lg n)^{-1/5}} \cdot (\lg n)^{O(\lg n)^{2/5}} = m^{1+o(1)}$.
\end{proof}

We now prove \Cref{lem:partition} using \Cref{lem:subset-U}. Begin with $U=V_{\ge d}$ and repeatedly apply \Cref{lem:subset-U} to obtain disjoint $X_1,\ldots,X_k\subseteq U$, and then reassign $U$ to be $U\setminus\bigcup_{i\in[k]}X_i$ for the next iteration; stop when $|U|=1$. Since the size of $U$ halves at each iteration, the number of iterations is at most $\lceil\log_2n\rceil$. We thus obtain $\lceil\log_2n\rceil\cdot 2n/d$ sets, each of which is $(d,\phi)$-well-linked in $G$, where $\phi=n^{-o(1)}$.

\section{Conclusion}

In this paper, we gave an $n^{2+o(1)}$-time algorithm for constructing 
a Gomory-Hu tree in a simple, undirected graph thereby solving the 
All Pairs Minimum Cuts problem in the same running time. 
Generalizing this result to weighted graphs, thereby improving on
Gomory and Hu's 60-year old algorithm that uses $n-1$ maxflow calls 
would be a breakthrough result. An intermediate goal would be to show 
this for unweighted multigraphs, i.e., allowing parallel edges but 
not edge weights. The $\tO(mn)$-time Gomory-Hu tree algorithms 
of Bhalgat~{\em et al.}~\cite{HariharanKPB07}
and of Karger and Levine~\cite{KargerL15} apply to these graphs, but not to
general weighted graphs, suggesting that this intermediate class might be easier
for the \apmc problem than general weighted graphs. Obtaining subcubic (in $n$) running times for 
the \apmc problem in unweighted (but not necessarily simple) graphs
remains an interesting open question.

A different question concerns the optimality of the result presented
in this paper. As
we discussed, our result is nearly optimal if mincut values have to be 
explicitly reported for all vertex pairs. Even if that is not required, 
our algorithm is nearly optimal if the input graph is dense, i.e., if
$m = \Theta(n^2)$. So, that leaves graphs containing 
$o(n^2)$ edges under the condition
that we do not need explicit reporting of mincut 
values for all vertex pairs. Ideally, for such graphs, one would like 
to design a near-linear time algorithm, i.e., a running time of $m^{1+o(1)}$.
But, that is not known even for a single $s$-$t$ mincut, i.e. for the 
maxflow problem. A more immediate goal is to construct a Gomory-Hu tree
via a subpolynomial (or polylogarithmic) number of maxflow calls. Indeed,
this was recently achieved at the cost of obtaining an 
{\em approximate} Gomory-Hu tree instead of an exact one~\cite{LiP21approximate}. 
For the exact problem, the current paper 
gives a reduction, but to polylogarithmic calls of the single source
mincut problem rather than the $s$-$t$ mincut problem.\footnote{\cite{AbboudKT20focs}
also give a similar reduction, although they require the oracle to actually report mincuts 
while we only require the mincut values.} Clearly, the 
former is a more powerful oracle, and hence the reduction is easier.
Improving this reduction to the $s$-$t$ mincut problem, or equivalently
removing the approximation in the result of \cite{LiP21approximate},
remains an interesting open question.

\bibliographystyle{alpha}
\bibliography{ref}

\newcommand{\etalchar}[1]{$^{#1}$}
\begin{thebibliography}{vdBLL{\etalchar{+}}21}

\bibitem[ACZ98]{ArikatiCZ98}
Srinivasa~Rao Arikati, Shiva Chaudhuri, and Christos~D. Zaroliagis.
\newblock All-pairs min-cut in sparse networks.
\newblock {\em J. Algorithms}, 29(1):82--110, 1998.

\bibitem[AKT20a]{AbboudKT20focs}
Amir Abboud, Robert Krauthgamer, and Ohad Trabelsi.
\newblock Cut-equivalent trees are optimal for min-cut queries.
\newblock In {\em 61st {IEEE} Annual Symposium on Foundations of Computer
  Science, {FOCS} 2020, Durham, NC, USA, November 16-19, 2020}, pages 105--118.
  {IEEE}, 2020.

\bibitem[AKT20b]{AbboudKT20soda}
Amir Abboud, Robert Krauthgamer, and Ohad Trabelsi.
\newblock New algorithms and lower bounds for all-pairs max-flow in undirected
  graphs.
\newblock In Shuchi Chawla, editor, {\em Proceedings of the 2020 {ACM-SIAM}
  Symposium on Discrete Algorithms, {SODA} 2020, Salt Lake City, UT, USA,
  January 5-8, 2020}, pages 48--61. {SIAM}, 2020.

\bibitem[AKT21]{AbboudKT20subcubic}
Amir Abboud, Robert Krauthgamer, and Ohad Trabelsi.
\newblock Subcubic algorithms for {Gomory-Hu} tree in unweighted graphs.
\newblock In {\em Proceedings of the 53rd Annual {ACM} Symposium on Theory of
  Computing}, 2021.

\bibitem[BENW16]{BorradaileENW16}
Glencora Borradaile, David Eppstein, Amir Nayyeri, and Christian
  Wulff{-}Nilsen.
\newblock All-pairs minimum cuts in near-linear time for surface-embedded
  graphs.
\newblock In S{\'{a}}ndor~P. Fekete and Anna Lubiw, editors, {\em 32nd
  International Symposium on Computational Geometry, SoCG 2016, June 14-18,
  2016, Boston, MA, {USA}}, volume~51 of {\em LIPIcs}, pages 22:1--22:16.
  Schloss Dagstuhl - Leibniz-Zentrum f{\"{u}}r Informatik, 2016.

\bibitem[BHKP07]{HariharanKPB07}
Anand Bhalgat, Ramesh Hariharan, Telikepalli Kavitha, and Debmalya Panigrahi.
\newblock {An {\~{O}}(mn) {G}omory-Hu tree construction algorithm for
  unweighted graphs}.
\newblock In {\em Proceedings of the 39th Annual {ACM} Symposium on Theory of
  Computing, San Diego, California, USA, June 11-13, 2007}, pages 605--614,
  2007.

\bibitem[BSW10]{BorradaileSW10}
Glencora Borradaile, Piotr Sankowski, and Christian Wulff{-}Nilsen.
\newblock Min st-cut oracle for planar graphs with near-linear preprocessing
  time.
\newblock In {\em 51th Annual {IEEE} Symposium on Foundations of Computer
  Science, {FOCS} 2010, October 23-26, 2010, Las Vegas, Nevada, {USA}}, pages
  601--610. {IEEE} Computer Society, 2010.

\bibitem[CGL{\etalchar{+}}20]{ChuzhoyGLNPS20}
Julia Chuzhoy, Yu~Gao, Jason Li, Danupon Nanongkai, Richard Peng, and
  Thatchaphol Saranurak.
\newblock A deterministic algorithm for balanced cut with applications to
  dynamic connectivity, flows, and beyond.
\newblock In {\em 61st {IEEE} Annual Symposium on Foundations of Computer
  Science, {FOCS} 2020, Durham, NC, USA, November 16-19, 2020}, pages
  1158--1167, 2020.

\bibitem[CQ21]{CQ20}
Chandra Chekuri and Kent Quanrud.
\newblock Isolating cuts, (bi-)submodularity, and faster algorithms for global
  connectivity problems.
\newblock {\em CoRR}, abs/2103.12908, 2021.

\bibitem[GH61]{GomoryH61}
Ralph~E Gomory and Tien~Chung Hu.
\newblock Multi-terminal network flows.
\newblock {\em Journal of the Society for Industrial and Applied Mathematics},
  9(4):551--570, 1961.

\bibitem[Gus90]{Gusfield90}
Dan Gusfield.
\newblock Very simple methods for all pairs network flow analysis.
\newblock {\em {SIAM} J. Comput.}, 19(1):143--155, 1990.

\bibitem[HKP07]{HariharanKP07}
Ramesh Hariharan, Telikepalli Kavitha, and Debmalya Panigrahi.
\newblock Efficient algorithms for computing all low \emph{s-t} edge
  connectivities and related problems.
\newblock In {\em Proceedings of the Eighteenth Annual {ACM-SIAM} Symposium on
  Discrete Algorithms, {SODA} 2007, New Orleans, Louisiana, USA, January 7-9,
  2007}, pages 127--136, 2007.

\bibitem[KL15]{KargerL15}
David~R. Karger and Matthew~S. Levine.
\newblock Fast augmenting paths by random sampling from residual graphs.
\newblock {\em {SIAM} J. Comput.}, 44(2):320--339, 2015.

\bibitem[KLS20]{KathuriaLS20}
Tarun Kathuria, Yang~P. Liu, and Aaron Sidford.
\newblock Unit capacity maxflow in almost
  {\textdollar}o(m{\^{}}\{4/3\}){\textdollar} time.
\newblock In {\em 61st {IEEE} Annual Symposium on Foundations of Computer
  Science, {FOCS} 2020, Durham, NC, USA, November 16-19, 2020}, pages 119--130.
  {IEEE}, 2020.

\bibitem[LP20]{LiP20deterministic}
Jason Li and Debmalya Panigrahi.
\newblock Deterministic min-cut in poly-logarithmic max-flows.
\newblock In {\em 61st {IEEE} Annual Symposium on Foundations of Computer
  Science, {FOCS} 2020}. {IEEE} Computer Society, 2020.

\bibitem[LP21]{LiP21approximate}
Jason Li and Debmalya Panigrahi.
\newblock Approximate {Gomory-Hu} tree is faster than $n-1$ max-flows.
\newblock In {\em Proceedings of the 53rd Annual {ACM} Symposium on Theory of
  Computing}, 2021.

\bibitem[LS21]{LiS21}
Jason Li and Thatchaphol Saranurak.
\newblock Deterministic weighted expander decomposition in almost-linear time,
  2021.
\newblock arXiv:2106.01567.

\bibitem[NI92]{NagamochiI92}
Hiroshi Nagamochi and Toshihide Ibaraki.
\newblock A linear-time algorithm for finding a sparse k-connected spanning
  subgraph of a k-connected graph.
\newblock {\em Algorithmica}, 7(5{\&}6):583--596, 1992.

\bibitem[Pan16]{Panigrahi16}
Debmalya Panigrahi.
\newblock {Gomory-Hu} trees.
\newblock In {\em Encyclopedia of Algorithms}, pages 858--861. 2016.

\bibitem[vdBLL{\etalchar{+}}21]{BrandLLSSSW21minimum}
Jan van~den Brand, Yin~Tat Lee, Yang~P. Liu, Thatchaphol Saranurak, Aaron
  Sidford, Zhao Song, and Di~Wang.
\newblock Minimum cost flows, mdps, and $\ell_1$-regression in nearly linear
  time for dense instances.
\newblock 2021.
\newblock arXiv:2101.05719.

\end{thebibliography}
\appendix

\section{Reducing Gomory-Hu Tree to Single-Source Mincut Values}
\label{sec:reduction}

The goal of this section is to prove \Cref{lem:reduction}. 
Let us first define the problem that the oracle solves, which we name \emph{$k$-bounded single source connectivity}, abbreviated as \ssc.

\begin{definition}[\ssc]
For a graph $G=(V,E)$, a terminal set $X$, and a source terminal $s\in X$, the output to  \ssc is the values $\min\{\mincut_G(s,v),k\}$ for all terminals $v\in X\setminus\{s\}$.
\end{definition}
Note that in we showed in \Cref{sec:isolate_well-linked} how to solve \ssc fast when $X$ is a well-linked set. Below, we will actually prove something stronger than  \Cref{lem:reduction} by relaxing the task of the oracle: instead of requiring the oracle to compute  \ssc, we only require the following \emph{verification} problem.

\begin{definition}[\sscv]
The input to  \sscv is a graph $G=(V,E)$, a terminal set $X$, a source terminal $s\in X$, and values $\tilde\lambda_v:v\in X\setminus\{s\}$ such that $\tilde\lambda_v\ge\min\{\mincut(s,v),k\}$. The task is to determine, for each vertex $v\in X\setminus\{s\}$, whether or not $\tilde\lambda_v=\min\{\mincut(s,v),k\}$.
\end{definition}

Clearly, if the oracle can compute  \ssc, then it can easily answer  \sscv. By focusing on \sscv instead of the original  \ssc, we hope to direct future efforts at tackling the former problem, which appears more tractable and is still powerful enough to solve the partial tree problem.

For the rest of this section, we prove the following lemma, which implies \Cref{lem:reduction} as discussed above.

\begin{lemma}\label{lem:reduction2}
For any vertex set $X\subseteq V$ and value $k\ge0$, there is a randomized algorithm that outputs a partial tree of $G$ that w.h.p., captures all mincuts separating $X$ of size at most $k$ and no mincuts of size more than $k$. It makes calls to max-flow and \sscv on graphs with a total of $\tilde{O}(n)$ vertices and $\tilde{O}(m)$ edges, and runs for $\tilde{O}(m)$ time outside of these calls.
\end{lemma}

\begin{remark} 
This reduction should be compared with the result by \cite{AbboudKT20focs} which reduces computing a partial tree to a similar oracle for single source mincuts.
The main difference is that their oracle must be able to \emph{list
edges} crossing $(s,v)$-mincuts for each $v\in X\setminus\{s\}$
but, outside the oracle calls, they do not need to call max-flow.
Our oracle is potentially easier to implement: we only need to \emph{verify
mincut values}, but our reduction needs to call max flow. Another
difference is that their reduction requires the oracle to run on weighted
graphs even if the input graph is unweighted, while our oracle only needs to run on unweighted graphs. 

Our reduction also holds for \emph{weighted} graphs (assuming an oracle for weighted graphs), so for completeness, we include the weighted case even though it is not needed for our main result. The only non-trivial modification is that, in the contracted graphs, we combine parallel edges into a edge with combined weights. We show using a different argument that we can still bound the total number of (combined) edges by $\tilde{O}(m)$ over all recursive instances.
\end{remark}

Before we present the proof of \Cref{lem:reduction2}, we state a corollary that can be handy (but we do not need it in this paper). It says, given an algorithm for \sscv, we only need to call it and max flow $\Otil(1)$ times to obtain the whole Gomory-Hu tree.
\begin{corollary}\label{cor:reduc_from_GHT}
Given a graph $G$ with $n$ vertices and $m$ edges,
there is a randomized algorithm that computes a Gomory-Hu tree of $G$ by making calls to max-flow and \sscv (for several different $k$'s) on graphs with a total of $\tilde{O}(n)$ vertices and $\tilde{O}(m)$ edges, and runs for $\tilde{O}(m)$ time outside of these calls.
\end{corollary}
The proof of \Cref{cor:reduc_from_GHT} is simply by calling \Cref{lem:reduction2} with $k = 2^i$ from $i =0$ to $O(\log n)$ to iteratively refine the partial tree until it captures all cut sizes, i.e., it becomes a Gomory-Hu tree. Note that this goes in a very similar way as in the proof of \Cref{lem:refineOnSparsifier}.

Now, we prove \Cref{lem:reduction2}. Our approach for proving \Cref{lem:reduction2} is almost identical to the one in \cite{LiP21approximate}, except we adapt their approximate Gomory-Hu tree algorithm to the exact case with the additional $k$-bounded property in mind. The algorithm is described in Algorithm~\ref{ghtree} a few pages down. Before we present Algorithm~\ref{ghtree}, we first consider the subprocedure \ref{step} that it uses, which mirrors the procedure \textsc{CutThresholdStep} from \cite{LiP21approximate}. Below, for any vertex set $S\subset V$, we define $\partial_G S = E_G(S,V\setminus S)$. 
\begin{algorithm}
\mylabel{step}{\textsc{PartialTreeStep}}\caption{\ref{step}$(G=(V,E),s,U,k)$} 
\begin{enumerate}
\item Initialize $R^0\gets U$ and $D\gets\emptyset$
\item For all $i$ from $0$ to $\lfloor\lg|U|\rfloor$ do:
 \begin{enumerate}
 \item Call \Cref{lem:isolating} on $T=R^i$, obtaining a $v$-minimal $(v,R^i\setminus v)$-mincut for each $v\in R^i$. Let $S^i_v$ be the side of the $(v,R^i\setminus v)$-mincut containing $v$ \label{line:Sv}
 \item Call  \sscv on graph $G$, terminal set $X=R^i$, source $s$, and values $\tilde\lambda_v=| \partial S^i_v |$ for $v\in R^i\setminus\{s\}$ 
 \item Let $D^i\subseteq R^i$ be the union of $S^i_v\cap U$ over all $v\in R^i\setminus\{s\}$ satisfying $\tilde\lambda_v=\min\{\mincut(s,v),k\}$ and $|S^i_v\cap U|\le|U|/2$\label{line:D}
 \item $R^{i+1}\gets$ subsample of $R^i$ where each vertex in $R^i\setminus \{s\}$ is sampled independently with probability $1/2$, and $s$ is sampled with probability $1$
 \end{enumerate}
\item Return the largest set $D^i$ and the corresponding sets $S^i_v$ over all $v\in R^i\setminus\{s\}$ satisfying the conditions in line~\ref{line:D} 
\end{enumerate}
\end{algorithm}

Let $D=D^0\cup D^1\cup \cdots\cup D^{\lfloor\lg|U|\rfloor}$ be the union of the sets $D^i$ as defined in Algorithm~\ref{step}. Let $D^*$ be all vertices $v\in U\setminus \{s\}$ for which there exists an $(s,v)$-mincut of size at most $k$ whose $v$ side has at most $|U|/2$ vertices in $U$. The lemma below is almost identical to Lemma~2.5 in \cite{LiP21approximate}; the only difference is that \textsc{CutThresholdStep} in \cite{LiP21approximate} focuses on solving what they call the Cut Threshold problem, whereas we tackle the partial Gomory-Hu tree problem directly.

\begin{lemma}\label{lem:step}
We have $D^i\subseteq D^*$ for all $i$. Moreover, the largest set $D^i$ returned by \ref{step} satisfies $\mathbb E[|D^i|] \ge \Omega(|D^*|/\log|U|)$.
\end{lemma}
\begin{proof}
We first prove that $D^i\subseteq D^*$ for all $i$. Each vertex $u\in D^i$ belongs to some $S^i_v$ satisfying $|\partial S^i_v|=\min\{\mincut(s,v),k\}\le k$ and $|S^i_v\cap U|\le |U|/2$. 
In particular, $\partial S^i_v$ is an $(s,u)$-mincut of size at most $ k$ whose side $S^i_v$ containing $u$ has at most $|U|/2$ vertices in $U$, so $u\in D^*$.

It remains to prove that $\mathbb E[|D^i|]\ge\Omega(|D^*|/\log|U|)$ for the largest set $D^i$.
For each vertex $v\in D^*$, let $S_v$ be the minimal $(v,s)$-mincut, and define $U_v=S_v\cap U$ and $n_v=|U_v|$. We say that a vertex $v\in D^*$ is \emph{active} if $v\in R^i$ for $i=\lfloor\lg n_{v}\rfloor$. In addition, if $U_v\cap R^i=\{v\}$, then we say that $v$ \emph{hits} all of the vertices in $U_v$ (including itself); see Figure~\ref{fig:hits}. In particular, in order for $v$ to hit any other vertex, it must be active. For completeness, we say that any vertex in $U\setminus D^*$ is not active and does not hit any vertex.

\begin{figure}\centering
\includegraphics[scale=1]{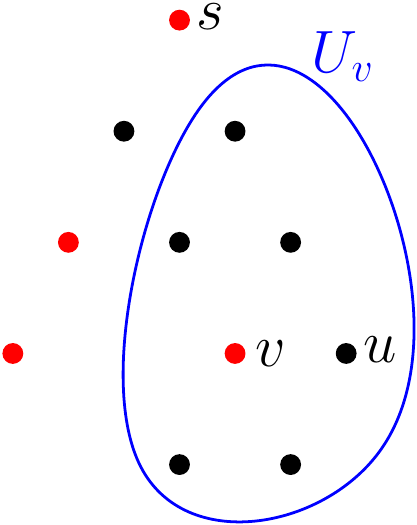}
\caption{Let $i=\lfloor\lg n_{v}\rfloor=\lfloor\lg 7\rfloor=2$, and let the red vertices be those sampled in $R^2$. Vertex $v$ is active and hits $u$ because $v$ is the only vertex in $U_{v}$ that is red.}\label{fig:hits}
\end{figure}

To prove that $\mathbb E[|D^i|] \ge \Omega(|D^*|/\log|U|)$, we will show that
 \begin{enumerate}
 \item[(a)] each vertex $u$ that is hit by some vertex $v$ is in $D^{\lfloor\lg n_v\rfloor}$, 
 \item[(b)] the total number of pairs $(u,v)$ for which $v\in D^*$ hits $u$ is at least $c |D^*|$ in expectation for some small enough constant $c>0$, and
 \item[(c)] each vertex $u$ is hit by at most $\lfloor\lg|U|\rfloor+1$ vertices. 
 \end{enumerate}

For (a), let $v$ be the vertex that hits $u$, and consider $i=\lfloor\lg n_{v}\rfloor$. We have $U_v\cap R^i=\{v\}$ by assumption, so $\partial S_v$ is a $(v,R^i\setminus\{v\})$-cut. On the other hand, we have that $\partial S^i_v$ is a $(v,R^i\setminus\{v\})$-mincut, so in particular, it is a $(v,s)$-cut. It follows that $\partial S_v$ and $\partial S^i_v$ are both $(v,s)$-mincuts and $(v,R^i\setminus v)$-mincuts, and $|\partial S^i_v|=\mincut(s,v)\le k$. Since $S_v$ is the minimal $(v,s)$-mincut and $S^i_v$ is a $(v,s)$-mincut, we must have $S_v \subseteq S^i_v$. Likewise, since $S_v$ is a $(v,R^i\setminus\{v\})$-mincut and $S^i_v$ is the minimal $(v,R^i\setminus\{v\})$-mincut, we also have $S^i_v\subseteq S_v $. It follows that $S_v=S^i_v$. Since $S_v$ is the minimal $(v,s)$-mincut and $v\in D^*$, we must have $|S_v\cap U|\le z$, so in particular, $|S^i_v\cap U|=|S_v\cap U|\le z$. Therefore, the vertex $v$ satisfies all the conditions of line~\ref{line:D}. Moreover, since $u\in U_v\subseteq S_v= S^i_v$, vertex $u$ is added to $D$ in the set $S^i_v\cap U$. 

For (b), for $i=\lfloor\lg n_{v}\rfloor$, we have $v\in R^i$ with probability exactly $1/2^i = \Theta(1/n_{v})$, and with probability $\Omega(1)$, no other vertex in $U_{v}$ joins $R^i$. Therefore, $v$ is active with probability $\Omega(1/n_{v})$. Conditioned on $v$ being active, it hits exactly $n_{v}$ many vertices. It follows that $v$ hits $\Omega(1)$ vertices in expectation. Summing over all $v\in D^*$ and applying linearity of expectation proves~(b).

For (c), since the isolating cuts $S^i_v$ over $v\in R^i$ are disjoint for each $i$, each vertex is hit at most once on each iteration $i$. Since there are $\lfloor\lg|U|\rfloor+1$ many iterations, the property follows.

Finally, we show why properties (a) to (c) imply $\mathbb E[|D^i|] \ge \Omega(|D^*|/\log|U|)$ for the largest $D^i$. By property~(b), the number of times some vertex hits another vertex is $\Omega(|D^*|)$ in expectation. Since there are $\lfloor\lg|U|\rfloor+1$ many distinct values of $\lfloor\lg n_v\rfloor$, there exists an integer $i$ for which the number of times some vertex $v$ with $\lfloor\lg n_v\rfloor=i$ hits another vertex is $\Omega(|D^*|/\log|U|)$ in expectation. Since each vertex is hit at most once on iteration $i$, there must be $\Omega(|D^*|/\log|U|)$ many vertices hit, all of which are included in $D^i$ by property~(a).
\end{proof}

\begin{algorithm}
\mylabel{ghtree}{\textsc{PartialTree}}
\caption{\ref{ghtree}$(G=(V,E),U,k)$ \label{alg:reduction}} 
\begin{enumerate}
\item Compute the Steiner connectivity $\lambda \gets \min_{u,v\in U} \mincut_G(u,v)$ w.r.t. terminals $U$ \\(If $|U|=1$, then $\lambda = \infty$) \\ \Comment{$\tilde O(|E|)$ time plus max-flow calls on graphs totalling $\Otil(|V|)$ vertices and $\Otil(|E|)$ edges~\cite{CQ20}}\label{line:steiner}
\item If $\lambda>k$, then terminate and return the trivial partial tree $(T,\mathcal P)$ with $V(T)=\{v\}$ for an arbitrary $v\in U$ and $\mathcal P=\{U\}$ as the trivial partition of $U$
\item $s\gets$ uniformly random vertex in $U$
\item Call $\ref{step}(G,s,U,k)$ to obtain $D^i$ and the sets $S^i_v$ (so that $D^i=\bigcup_v S^i_v\cap U$)
\ \label{line:max}
\item For each set $S^i_v$ do: \Comment{Construct recursive graphs and apply recursion}
 \begin{enumerate}
 \item Let $G_v$ be the graph $G$ with vertices $V\setminus S^i_v$ contracted to a single vertex $x_v$ \Comment{$S^i_v$ are disjoint}
 \item Let $U_v\gets S^i_v\cap U$
 \item If $|U_v|>1$, then recursively set $(T_v,\mathcal P_v)\gets\ref{ghtree}(G_v,U_v)$
 \end{enumerate}
\item Let $G_\lar$ be the graph $G$ with (disjoint) vertex sets $S^i_v$ contracted to single vertices $y_v$ for all $v\in D^i$
\item Let $U_\lar\gets U\setminus D^i$
\item If $|U_\lar|>1$, then recursively set $(T_\lar,\mathcal P_\lar)\gets\ref{ghtree}(G_\lar,U_\lar)$

\item Combine $(T_\lar,\mathcal P_\lar)$ and $\{(T_v,\mathcal P_v):v\in D^i\}$ into $(T,\mathcal P)$ according to \ref{combine}

\item Return $(T,\mathcal P)$

\end{enumerate}
\end{algorithm}

\begin{algorithm}
\mylabel{combine}{\textsc{Combine}}\caption{\ref{combine}$((T_\lar,\mathcal P_\lar),\{(T_v,\mathcal P_v): v\in R^i\} )$} \begin{enumerate}
    \item Construct $T$ by starting with the disjoint union $T_\lar\cup\bigcup_{v\in R^i}T_v$ and, for each $v\in R^i$, adding an edge $(x,y)$ of weight $| \partial_GS^i_v |$, where $x_v\in V_x$ and $y_v\in V_y$ 
    \label{line:combine-T}
    \item Construct $\mathcal P$ as the disjoint union of partitions $\mathcal P_\lar$ and $\mathcal P_v$ over all $v\in R^i$, restricted to vertices in $V$ 
    \label{line:combine-f}
    \item Return $(T,\mathcal P)$
\end{enumerate}

\end{algorithm}

Now, we state \Cref{alg:reduction} for \Cref{lem:reduction2} and prove its correctness below.

\begin{lemma}[Correctness]\label{lem:correctness}
For any vertex set $U\subseteq V$ and value $k\ge0$, the algorithm $\ref{ghtree}(G,U,k)$ returns a partial tree $(T,\cP)$ in $G$ with the terminal set $V(T) \subseteq U$ that captures all mincuts separating $U$ of size at most $k$ 
and no mincuts of size more than $k$.
\end{lemma}
We give a detailed proof of \Cref{lem:correctness} in \Cref{sec:reduction_correct_proof} because it follows using a standard
argument. Here, we give only a high-level argument: although
in each recursion level of \Cref{alg:reduction}, the algorithm refines the tree into
many parts, these refinements can be simulated by a sequence of the
standard refinement steps in the original Gomory-Hu tree algorithm
that splits only one supernode/part into two. So the resulting $(T,\cP)$
is indeed a partial tree of $G$. Next, $(T,\cP)$ captures no mincut of
size more than $k$ because all edges in $T$ have weight at most $k$ by construction. 
It remains to argue why $(T,\cP)$ captures all
mincuts of size at most $k$ separating terminals $U$, let $x,y\in U$
where $\mincut_{G}(x,y)\le k$. We want to say that $x,y$ are in
different parts of $\cP$. There are two main cases. If both $x,y\in U_{\lar}$
or both $x,y\in U_{v}$ for some $v\in D^{i}$, then, by induction, $x$ are $y$ will
be separated in $\cP_{\lar}$ or in $\cP_{v}$ respectively, and they
remain separated in $\cP$ by how the \textsc{Combine} subroutine works. Otherwise,
we have that $x\in U_{v}$ and $y\notin U\setminus U_{v}$, then they
are in different subproblems in the recursion and remain separated
in $\cP$ by how the \textsc{Combine} subroutine works.

The remaining part of this section is for bounding the running time of \Cref{alg:reduction}.
For any $U\subseteq V$, define $f_k(U)$ as the size of the largest subset of vertices in $U$ whose pairwise mincut values are all greater than $k$. The following lemma is inspired by~\cite{AbboudKT20focs}. (In fact, our statement and proof are identical to theirs in the case $k=\infty$.)
\begin{claim}\label{lem:ghtree:random-s}
Let the vertex $s\in U$ be chosen uniformly at random. Then, $\mathbb E[|D^*|]=(|U|-f_k(U))/2$.
\end{claim}
\begin{proof}
Consider the partial tree $(T,\mathcal P)$ that captures all mincuts separating $U$ of size at most $k$ and no mincuts of size more than $k$. In other words, vertices $x,y\in U$ belonging to different parts of $\mathcal P$ iff $\mincut(x,y)\le k$. By the definition of $f_k(U)$, the maximum size of a part in $\mathcal P$ is $f_k(U)$. Consider a digraph on vertex set $U$ where for each pair of vertices $x,y\in U$, we add a directed edge $(x,y)$ if there exists an $(x,y)$-mincut of size at most $k$ where the side containing $x$ has at most $|U|/2$ vertices in $U$. Clearly, for each $x,y\in U$ with $\mincut(x,y)\le k$, we add either $(x,y)$ or $(y,x)$ (or both) to the digraph. Also, for each vertex $x\in U$, the vertices $y\in U$ with $\mincut(x,y)\le k$ are exactly those not in the same part as $u$ in $\mathcal P$, so there are at least $|U|-f_k(U)$ such vertices $y$. Therefore, the total number of arcs entering or leaving $u$ is at least $|U|-f_k(U)$. The total number of arcs in the digraph is at least $|U|(|U|-f_k(U))/2$, so the average out-degree is at least $(|U|-f_k(U))/2$. Note that for the vertex $s$ chosen uniformly at random, the set $|D^*|$ is exactly the out-neighbors of $s$. It follows that $\mathbb E[|D^*|]\ge(|U|-f_k(U))/2$.
\end{proof}

\begin{lemma}\label{lem:depth}
W.h.p., the algorithm \ref{ghtree} has maximum recursion depth $\polylog(n)$.
\end{lemma}
\begin{proof}
By construction, each recursive instance $(G_v,U_v)$ has $|U_v|\le|U|/2$.

By \Cref{lem:step} and \Cref{lem:ghtree:random-s}, over the randomness of $s$ and \ref{step}, we have
\[ \mathbb E[D^i]\ge \Omega\left(\frac{\mathbb E[|D^*|]}{\log|U|}\right) \ge \Omega\left(\frac{|U|-f_k(U))}{\log|U|}\right) .\]
If $f_k(U)\le\frac23|U|$, then this is $\Omega(|U|/\log|U|)$, so the recursive instance $(G_\lar,U_\lar)$ satisfies $\mathbb E[|U_\lar|]\le(1-1/\log|U|)\cdot|U|$. 
Suppose now that $f_k(U)\ge\frac23|U|$, and let $U'\subseteq U$ be the vertex set of size $f_k(U)$ whose pairwise mincut values exceed $k$. By construction, the entire set $U'$ is contained in either $U_\lar$ or $U_v$ for some $v\in D^i$, and since $|U_v|\le|U|/2$ for all $v\in D^i$, we must have $U'\subseteq U_\lar$. In other words, $f_k(U_\lar)=f_k(U)$. So the recursive instance $(G_\lar,U_\lar)$ satisfies
\begin{align*}
     \mathbb E[|U_\lar|-f_k(U_\lar)] = \mathbb E[|U_\lar|]-f_k(U) &\ge |U|-\Omega\left(\frac{|U|-f_k(U)}{\log|U|}\right)-f_k(U) \\&= \left(1-\Omega\bigg(\frac1{\log|U|}\bigg)\right)(|U|-f_k(U)) .
\end{align*}

Therefore, each recursive branch has either $|U|-f_k(U)$ dropped by factor $(1-\Omega(1/\log|U|))$ in expectation, or $|U|$ dropped by factor $(1-\Omega(1/\log|U|))$ in expectation (and $|U|-f_k(U)$ can potentially increase\footnote{Actually $|U|-f_k(U)$ cannot increase, but we do not prove it since it is unnecessary.}). It follows that w.h.p., all branches reach $|U|=f_k(U)$ or $|U|=1$ by $\polylog(n)$ recursive calls. In both cases, the value $\lambda$ in line~\ref{line:steiner} has $\lambda>k$, so the algorithm terminates.
\end{proof}

\begin{lemma}[Running time]\label{lem:ghtree:runtime}
For an unweighted (respectively, weighted) graph $G=(V,E)$, and terminals $U\subseteq V$, and value $k\ge0$, $\ref{ghtree}(G,U,k)$ takes time $\tilde{O}(m)$ plus calls to \sscv and max-flow on unweighted (respecively, weighted) instances with a total of $\tilde{O}(n)$ vertices and $\tilde{O}(m)$ edges.
\end{lemma}
\begin{proof}
For a given recursion level, consider the instances $\{ (G_i,U_i,k)\}_i$ across that level. By construction, the terminals $U_i$ partition $U$. Moreover, the total number of vertices over all $G_i$ is at most $n+2(|U|-1)=O(n)$ since each branch creates $2$ new vertices and there are at most $|U|-1$ branches. 

To bound the total number of edges, we consider the unweighted and weighted cases separately, starting with the unweighted case. The total number of new edges created is at most the sum of weights of the edges in the final Gomory-Hu Steiner tree. For an unweighted graph, this is $O(m)$ by the following well-known argument. Root the Gomory-Hu Steiner tree $T$ at any vertex $r\in U$; for any $v\in U\setminus r$ with parent $u$, the cut $\partial\{v\}$ in $G$ is a $(u,v)$-cut of value $\deg(v)$, so $w_T(u,v)\le\lambda_G(u,v)\le\deg(v)$. Overall, the sum of the edge weights in $T$ is at most $\sum_{v\in U}\deg(v)\le2m$.

For the weighted case, define a \emph{parent} vertex in an instance as a vertex resulting from either (1)~contracting $V\setminus S_v$ in some previous recursive $G_v$ call, or (2)~contracting a component containing a parent vertex in some previous recursive call. There are at most $O(\log n)$ parent vertices: at most $O(\log n)$ can be created by~(1) since each $G_v$ call decreases $|U|$ by a constant factor, and (2)~cannot increase the number of parent vertices. Therefore, the total number of edges adjacent to parent vertices is at most $O(\log n)$ times the number of vertices. Since there are $O(n)$ vertices in a given recursion level, the total number of edges adjacent to parent vertices is $O(n\log n)$ in this level. Next, we bound the number of edges not adjacent to a parent vertex by $m$. To do so, we first show that on each instance, the total number of these edges over all recursive calls produced by this instance is at most the total number of such edges in this instance. Let $P\subseteq V$ be the parent vertices; then, each $G_v$ call has exactly $|E(G[S_v\setminus P])|$ edges not adjacent to parent vertices (in the recursive instance), and the $G_\lar$ call has at most $|E(G[V\setminus P]) \setminus \bigcup_{v\in R}E(G[S_v\setminus P])|$, and these sum to $|E(G[V\setminus P])|$, as promised. This implies that the total number of edges not adjacent to a parent vertex at the next level is at most the total number at the previous level. Since the total number at the first level is $m$, the bound follows.

Therefore, there are $O(n)$ vertices and $\tO(m)$ edges in each recursion level. By \Cref{lem:depth}, there are $\text{polylog}(n)$ levels, for a total of $\tilde{O}(n)$ vertices and $\tilde{O}(m)$ edges. In particular, the instances to the max-flow calls have $\tilde{O}(n)$ vertices and $\tilde{O}(m)$ edges in total.
\end{proof}

\subsection{Proof of \Cref{lem:correctness}}
\label{sec:reduction_correct_proof}

To prove \Cref{lem:correctness}, we first introduce a helper proposition, which follows from the standard argument on non-crossing cuts used in the original Gomory-Hu algorithm. We include the proof only for completeness.

\begin{prop}\label{lem:ghtree:exact}
For any distinct vertices $p,q\in U_\lar$, we have $\mincut_{G_\lar}(p,q) = \mincut_G(p,q)$. The same holds with $U_\lar$ and $G_\lar$ replaced by $U_v$ and $G_v$ for any $v\in D^i$.
\end{prop}
\begin{proof}
Since $G_\lar$ is a contraction of $G$, we have $\mincut_{G_\lar}(p,q) \ge \mincut_G(p,q)$. To show the reverse inequality, fix any $(p,q)$-mincut in $G$, and let $S$ be one side of the mincut. We show that for each $v\in  D^i$, either $S^i_v\subseteq S$ or $S^i_v\subseteq V\setminus S$. Assuming this, the cut $\partial_G S$ stays intact when the sets $S^i_v$ are contracted to form $G_\lar$, so $\mincut_{G_\lar}(p,q) \le |\partial_G S| = \mincut_G(p,q)$.

Consider any $v\in D^i$, and suppose first that $v\in S$. Then, $S^i_v\cap S$ is still a $(v,R^i\setminus v)$-cut, and $S^i_v\cup S$ is still a $(p,q)$-cut. By the submodularity of cuts,
\[ |\partial_GS^i_v| + |\partial_GS| \ge |\partial_G(S^i_v\cup S)| + |\partial_G(S^i_v\cap S)|. \]
In particular, $S^i_v\cap S$ must be a $(v,R^i\setminus v)$-mincut. Since $S^i_v$ is the $v$-minimal $(v,R^i\setminus v)$-mincut by \Cref{lem:isolating} called in subprocedure \ref{step}, it follows that $S^i_v\cap S = S^i_v$, or equivalently, $S^i_v\subseteq S$.

Suppose now that $v\notin S$. In this case, we can swap $p$ and $q$, and swap $S$ and $V\setminus S$, and repeat the above argument to get $S^i_v\subseteq V\setminus S$.

The argument for $U_v$ and $G_v$ is identical, and we skip the details.
\end{proof}

\begin{proof}[Proof (\Cref{lem:correctness}).]
By construction, all edges in $T$ have weight at most $k$ (i.e.,~it captures no mincut of size more than $k$). It remains to show that it captures all mincuts separating $U$ of size at most $k$. That is, for all $x,y\in U$ with $\mincut(x,y)\le k$, there is an edge on the $x$-$y$ path in $T$ of weight $\mincut(x,y)$.

We apply induction on $|U|$. 
By induction, the recursive outputs $(T_\lar,\mathcal P_\lar)$ and $(T_v,\mathcal P_v)$ are partial trees capturing all mincuts separating $U_\lar$ in $G_\lar$ and $U_v$ in $G_v$, respectively, of size at most $k$.

First, consider $x,y\in U_\lar$ with $\mincut(x,y)\le k$, so that the partition $\mathcal P_\lar$ separates $x$ and $y$. Let $(u,u')$ be the minimum-weight edge on the $x$-$y$ path in $T_\lar$, and let $U'_\lar\subseteq U_\lar$ be the vertices of the connected component of $T_\lar \setminus (u,u')$ containing $x$, so that $V_{U'_\lar}$ is an $(s,t)$-mincut in $G_\lar$ with value $w_T(u,u')$. Define $U'\subseteq U$ as the vertices of the connected component of $T \setminus (u,u')$ containing $x$. By construction of $(T,\mathcal P)$,  
$G_\lar$ is simply $G$ with all vertex sets $S^i_v$ contracted to $y_v$ for all $v\in D^i$. 
Similarly, 
$V_{U'_\lar}$ (union of parts $V_z:z\in U'_\lar$ from $\mathcal P_\lar$) is simply
the set $V_{U'}$ (union of parts $V_z:z\in U'$ from $\mathcal P$) where all vertex sets $S^i_v$ are contracted to $y_v$ for all $v\in D^i$.
So we conclude that $w_{G_\lar}(V_{U'_\lar}) = w_G(V_{U'})$. By \Cref{lem:ghtree:exact}, we have $\mincut_G(x,y)=\mincut_{G_\lar}(x,y)$ are equal, so $w_G(V_{U'})$ is an $(x,y)$-mincut in $G$. In other words, the partial tree condition for $(T,\mathcal P)$ is satisfied for all $x,y\in U_\lar$ with $\mincut(x,y)\le k$. A similar argument handles the case $x,y\in U_v$ with $\mincut(x,y)\le k$ for some $v\in D^i$.

Consider now $x,y\in U$ with $\mincut(x,y)\le k$, and either $x\in U_v$ and $y\in U_\lar$, or $x\in U_v$ and $y\in U_{v'}$ for distinct $v,v'\in D^i$. Suppose first that $x\in U_v$ and $y\in U_\lar$. By considering which sides $v$ and $s$ lie on the $(x,y)$-mincut $(S,V\setminus S)$, we have
\[ |\partial_GS|=\mincut(x,y)\ge\min\{\mincut(x,v),\mincut(v,s),\mincut(s,y)\} .\]
We case on which of the three mincut values $\mincut(x,y)$ is greater than or equal to. 

\begin{enumerate}
\item If $\mincut(x,y)\ge\mincut(v,s)$, then since $S^i_v$ is a $(v,s)$-mincut that is also an $(x,y)$-cut, we have $\mincut(x,y)=\mincut(v,s)\le k$. By construction, there is an edge $e$ of weight $|\partial_G S^i_v|=|\partial_G S|$ on the $x$-$y$ path in $T$. There cannot be edges on the $x$-$y$ path in $T$ of smaller weight, since each edge corresponds to a $(x,y)$-cut in $G$ of the same weight. Therefore, $e$ is the minimum-weight edge on the $x$-$y$ path in $T$.\label{ghtree:case1}
\item Suppose now that $\mincut(x,v)\le \mincut(x,y)<\mincut(v,s)$. Let $z\in U_v$ be the vertex with $x_v\in V_z$ (for partition $\mathcal P_v$). Since $\mincut(x,v)\le\mincut(x,y)\le k$, the vertices $x,v$ are separated by the partition $\mathcal P_v$, and the minimum-weight edge $e$ on the $x-v$ path in $T_v$ has weight $\mincut(x,v)$. This edge $e$ cannot be on the $v-z$ path in $T_v$, since otherwise, we would obtain a $(v,x_v)$-cut of value $\mincut(x,v)$ in $G_v$, which becomes a $(v,s)$-cut in $G$ after expanding the contracted vertex $x_v$; this contradicts our assumption that $\mincut(x,v)<\mincut(v,s)$. It follows that $e$ is on the $x-z$ path in $T_v$ which, by construction, is also on the $x-y$ path in $T$. Once again, the $x-y$ path cannot contain an edge of smaller weight. \label{ghtree:case2}
\item The final case $\mincut(s,y)\le\mincut(x,y)<\mincut(v,s)$ is symmetric to case~\ref{ghtree:case2}, except we argue on $T_\lar$ and $G_\lar$ instead of $T_v$ and $G_v$.
\end{enumerate}

Suppose now that $x\in U_v$ and $y\in U_{v'}$ for distinct $v,v'\in D^i$. By considering which sides $v,v',s$ lie on the $(x,y)$-mincut, we have
\[ |\partial_GS|=\mincut(x,y)\ge\min\{\mincut(x,v),\mincut(v,s),\mincut(s,v'),\mincut(v',y)\} .\]
We now case on which of the four mincut values $\mincut(x,y)$ is greater than or equal to. 
\begin{enumerate}
\item If $\mincut(x,y)\ge\mincut(v,s)$ or $\mincut(x,y)\ge\mincut(s,v')$, then the argument is the same as case~\ref{ghtree:case1} above.
\item If $\mincut(x,v)\le \mincut(x,y)<\mincut(v,s)$ or $\mincut(y,v')\le\mincut(x,y)<\mincut(v',s)$, then the argument is the same as case~\ref{ghtree:case2} above.
\end{enumerate}
This concludes all cases, and hence the proof.
\end{proof}

\end{document}